\documentclass[lettersize,journal]{IEEEtran}
\usepackage{amsmath,amsfonts}
\usepackage{algorithmic}
\usepackage{algorithm}
\usepackage{array}
\usepackage[caption=false,font=normalsize,labelfont=sf,textfont=sf]{subfig}
\usepackage{textcomp}
\usepackage{stfloats}
\usepackage{url}
\usepackage{verbatim}
\usepackage{graphicx}
\usepackage[backend=bibtex,style=ieee]{biblatex}
\usepackage{amssymb}
\usepackage{xcolor}
\usepackage{bm}
\usepackage{latexsym}
\usepackage{multirow}
\usepackage{threeparttable}
\usepackage{amsthm}

\newtheorem{theorem}{Theorem}

\newtheorem{assumption}{Assumption}
\newtheorem{definition}{Definition}
\newtheorem{corollary}{Corollary}
\begin{document}

\title{Computationally Efficient Sparse Signal Recovery\\ via Linear Sketching and Deep Unfolding}

\author{Tatsuki Tokumura,~\IEEEmembership{Non-Member,} 
        \thanks{T. Tokumura, A. Nakai-Kasai, and T. Wadayama are with Nagoya Institute of Technology, Gokiso, Nagoya, Aichi 466-8555, Japan,}
        Ayano Nakai-Kasai,~\IEEEmembership{Member,~IEEE,} 
      and Tadashi Wadayama,~\IEEEmembership{Member,~IEEE} 
    %   \thanks{T. Wadayama is with Nagoya Institute of Technology, Gokiso, Nagoya, Aichi 466-8555, Japan,}
        % <-this % stops a space
\thanks{Part of this research was presented at the 2025 Asia Pacific Signal and Information Processing Association Annual Summit and Conference (APSIPA ASC) \cite{apsipatokumura}.}% <-this % stops a space
}

\markboth{Journal of \LaTeX\ Class Files,~Vol.~14, No.~8, August~2021}%
{Shell \MakeLowercase{\textit{et al.}}: A Sample Article Using IEEEtran.cls for IEEE Journals}


\maketitle

\begin{abstract}
  This paper provides a sparse signal recovery algorithm, 
  % In this paper, we propose a novel sparse signal recovery algorithm, 
  DU-PSISTA (Deep Unfolded-Periodic Sketched Iterative Shrinkage-Thresholding Algorithm), 
  which aims to balance computational efficiency
  and accuracy for recovering high-dimensional sparse signals, 
  and a convergence analysis under sufficient conditions. 
  DU-PSISTA introduces a random matrix projection known as sketching to reduce the dimensionality of gradient computations 
  and periodically 
  alternates between the standard ISTA and the sketched variant.
  % that leverages random projections to
  % reduce the dimensionality of gradient computations. 
  This hybrid structure enables flexible 
  control over the trade-off between accuracy and computational complexity through a pre-configurable period parameter. 
  The algorithm includes many parameters to be tuned such as step sizes and thresholding factors 
  so that we incorporate deep unfolding 
  % into the algorithm that embeds 
  % learnable parameters such as step sizes and thresholding factors into each iteration. 
  % These parameters are optimized through data-driven training
  that optimizes the parameters through data-driven training, enabling the algorithm to 
  adaptively improve convergence speed and performance. 
  We show that the proposed method achieves a linear-type contraction to a neighborhood of the true sparse signal 
  with properly selected parameters.
  The analysis provides an interpretation for the effectiveness of the hybrid structure to improve recovery accuracy.
  Numerical experiments confirm that 
  our method achieves comparable recovery performance to conventional deep unfolded ISTA
  while reducing computational complexity, especially when the period parameter and sketch 
  size are properly selected.
  The results are also consistent with the theoretical insights.
\end{abstract}


\section{Introduction}

\IEEEPARstart{D}{emand} for high-dimensional data reconstruction has been increasing 
with the recent advances in sensing technologies.
Reperesentative examples are medical image reconstruction with high temporal resolution 
and signal reconstruction acquired by massive number of IoT (Internet of Things) devices~\cite{woodruff2014sketching}.
A processing technique utilizing the sparsity of signal known as compressed sensing \cite{donoho2006compressed} has capability of efficient and accurate recovery of such high-dimensional data.
% \IEEEPARstart{R}{ecent} advances in sensing technologies, wireless communication systems, medical imaging, and IoT (Internet of Things)
% devices have led to a dramatic increase in both the dimensionality and volume of data acquired from the real world. For
% instance, fMRI (functional Magnetic Resonance Imaging) for brain activity mapping typically requires the acquisition of
% tens of thousands of voxel values at high temporal resolution, resulting in massive, high-dimensional datasets. Similarly,
% in applications such as wireless communications, radar systems, and image processing, acquired signals are frequently 
% high-dimensional and often include noise or missing entries. These challenges necessitate the development of signal
% processing methods capable of efficient and accurate recovery of such high-dimensional data~\cite{woodruff2014sketching}.
% In recent years, a sparse signal processing technique known as compressed sensing~\cite{eldar2012compressed,assyuku} has attracted
% considerable attention. 
The sparsity of signals appears in various forms 
such as in medical imaging, audio processing, communications, and robotics, 
and is utilized in a wide range of fields \cite{eldar2012compressed,assyuku}.

One example of the demand for high-dimensional sparse signal reconstruction is found in next-generation wireless networks. 
With the emergence of 
6G (6th Generation Mobile Communication System) and IoT technologies, the number of connected 
devices is growing exponentially, leading to a corresponding increase in the volume of 
signals communicated between base stations and devices \cite{prasad2025toward}. 
High accuracy and real-time capability must be achieved simultaneously 
for fundamental tasks in wireless communications such as channel estimation and active user detection 
to meet the demands for ultra-high-speed, high-capacity communication in 6G. 
In uplink scenarios, 
grant-free massive access without any scheduling in advance enables avoiding complex protocols for device scheduling, 
and the activity patterns of devices are temporally and spatially sparse \cite{liu2018massive,gao2024compressive}.
% Especially, in grant-free
% uplink scenarios~\cite{liu2018massive,gao2024compressive}, 
% where base station scheduling is eliminated, 
% the activity patterns of devices exhibit temporal and spatial sparsity. 
The problem of active user detection can be formulated as a sparse signal recovery problem
to be solved in real-time speed. 

ISTA (Iterative Shrinkage-Thresholding Algorithm)~\cite{charnbolle1998nonlinear,daubechies2004iterative} is a widely adopted iterative optimization
method for solving such a problem due to its simplicity and theoretical convergence guarantees. 
ISTA is known to have slow convergence speed with sublinear rate \cite{beck2009fast} 
but it can be incorporated into deep unfolded frameworks to improve performance.
Deep unfolding \cite{deep,deepmodel} is one of the hyperparameter optimization methods for iterative algorithms,
where the algorithm iterations are unfolded to form a feedforward network 
and the hyperparameters included in the algorithm are optimized by data-driven training~\cite{balatsoukas2019deep}.
Unlike deep neural networks, deep unfolding models maintain interpretability and 
generally require fewer training samples.
This idea was first introduced by Gregor and LeCun \cite{deep} as LISTA (Learned ISTA) to improve the convergence speed and performance of ISTA.
It has been reported the success of deep unfolding in applications such as MIMO (Multiple-Input Multiple-Output) signal detection and sparse signal estimation ~\cite{song2021soft,deepISTA,ito2019trainable,takabe2019trainable}, 
or in other algorithms beyond ISTA \cite{deng2014tutorial,yang2016deep,kishida2020deep}.

ISTA and its deep unfolded variants, however, include 
matrix-vector product in each algorithm iteration, 
which dominates the computational complexity of the total algorithm execution time \cite{ito2019trainable}.
% This results in high computational cost when dealing with large-scale data matrices. 
In large-scale systems such as massive MIMO or high-resolution imaging, 
the matrix-vector product complexity per iteration becomes a critical bottleneck for meeting strict real-time and low-latency requirements 
and is required to be reduced as much as possible \cite{balatsoukas2019deep,shlezinger2023model}.
To address this limitation, we focus on recent progress in sketching techniques~\cite{tang2017exploiting,tang2019randomized} to reduce the 
computational burden of iterative algorithms. 
Sketching projects high-dimensional data onto a lower-dimensional subspace 
using random projections while approximately preserving the original problem structure, 
limiting the scope to over-constrained systems. 
Algorithms such as IHS (Iterative Hessian Sketch) 
and GPIS (Gradient Projection Iterative Sketch) have demonstrated the ability to approximate solutions to 
the original problem using significantly fewer measurements, with theoretical guarantees on 
performance~\cite{woodruff2014sketching}. 
% However, the dimensionality reduction may introduce approximation errors, which 
% can in turn lead to a decrease in convergence speed.
The sketching techniques have been applied to various applications such as massive MIMO \cite{choi2021large,li2021fast}.
However, there is a trade-off between reducible dimension and convergence performance 
due to the approximated structure of sketching.
It is desirable to develop a method that achieves both high computational efficiency and reasonable recovery performance 
when applied to ISTA.

In this paper, we provide a hybrid framework that leverages both original updates and sketched updates of ISTA. 
The proposed algorithm alternates the original and sketched updates in a periodic manner to reduce the 
computational load associated with the matrix-vector product calculation. 
This algorithm design is expected to suppress deterioration 
of convergence performance caused by sketching. 
This hybrid architecture has a potential to flexibly control a trade-off 
between computational efficiency and recovery accuracy, 
which can be adjusted according to system requirements and environmental constraints.
In under-determined systems such as compressed sensing, 
the approximated preservation of the original problem structure via sketching is not guaranteed, 
and even when periodic structures are incorporated, concerns about performance degradation remain.
% The proposed algorithm has a number of tunable parameters that can affect performance 
% but it is not trivial to set appropriate values for all of them, similar to the issues encountered with ISTA.
We therefore incorporate the concept of deep unfolding into the proposed algorithm 
to improve performance through data-driven hyperparameter learning.
% , which interprets 
% iterative optimization algorithms as feedforward neural networks, and enables optimization of embedded learnable 
% parameters by data-driven training~\cite{balatsoukas2019deep}. 
%  due to their algorithm-driven structure.
% This approach builds on Gregor and LeCun's work improving ISTA and demonstrating its effectiveness, with reported success 
% in applications such as MIMO (Multiple-Input Multiple-Output) signal detection and sparse signal estimation ~\cite{deep,song2021soft,deepISTA}.

This paper not only verifies the method but also presents convergence analysis.
Theoretical analysis for deep unfolded algorithms are developing recently.
Chen et al. first theoretically analyzed the convergence of LISTA 
and have made it possible to explain why the introduction of learnable parameters in deep unfolding improves convergence \cite{chen2018theoretical}.
Moreover, a method for reducing learnable parameters based on theoretical analysis has also been reported \cite{liu2018alista}.
There also exists a framework that views deep unfolded algorithms, not limited to LISTA, as fixed-point iterations and discusses their convergence properties \cite{monga2021algorithm}.
Based on these insights, we theoretically analyze the convergence of the proposed algorithm.
The proposed method can achieve linear convergence to a neighborhood of the true sparse signal 
with properly selected parameters.
Furthermore, this analysis provides an explanation for the effectiveness of mixing the original and sketched updates 
instead of using only sketched updates to improve recovery accuracy.
% This implies that the proposed algorithm has possibility to improve the convergence rate of ISTA to a degree comparable to deep unfolded one such as LISTA.
This implies that the proposed DU-PSISTA retains the linear convergence property of deep unfolded networks like LISTA \cite{deep}, 
while significantly reducing the per-iteration computational cost through periodic sketching.

The goal of this work is to develop a theoretically supported novel sparse signal recovery algorithm that integrates sketching for computationally 
efficient operations and deep unfolding for performance improvement, thereby enabling efficient and scalable recovery of 
high-dimensional sparse signals in practical applications.
The main contributions of this paper are summarized as follows:
\begin{itemize}
  \item We propose a novel sparse signal recovery algorithm, DU-PSISTA (Deep Unfolded-Periodic Sketched ISTA), that periodically alternates between standard update and a sketched update for computationally efficient operations 
    and applies deep unfolding to make the parameters in each iteration learnable.
  \item We theoretically prove the linear convergence to a neighborhood of the true sparse signal of the proposed algorithm under certain conditions. 
    Moreover, we provide an explanation for the effectiveness of mixing the original and sketched updates to improve recovery accuracy.
  \item We verify the effectiveness of the proposed method with two types of sketching matrices through numerical experiments.
\end{itemize}

This paper extends our previous work \cite{apsipatokumura} 
that proposed DU-PSISTA algorithm.
The main differences from the previous work are as follows:
\begin{itemize}
  \item addition of convergence analysis,
  \item addition of detailed experimental evaluations in the recovery performance and the execution time,
  \item consideration of variations of matrices used for sketching (specifically, Count Sketch),
  \item and discussion based on a comparison of the learned parameters in the experiments and the convergence analysis.
\end{itemize}

The organization of this paper is as follows:
In Section \ref{sec:preliminaries}, we introduce the preliminaries of the proposed method.
In Section \ref{sec:proposed_method}, we propose the proposed method.
In Section \ref{sec:convergence_analysis}, we theoretically analyze the convergence of the proposed method.
In Section \ref{sec:experiments}, we verify the effectiveness of the proposed method through numerical experiments.
In Section \ref{sec:conclusion}, we conclude the paper.

\section{Preliminaries}
\label{sec:preliminaries}

\subsection{ISTA}
% Compressed sensing~\cite{donoho2006compressed} is the problem of estimating an unknown sparse original signal vector based on the 
% observation vector and a known observation matrix. A sparse vector is a vector in which the number of nonzero elements is 
% small relative to the length of the vector. This problem is defined by the following observation model. 
Consider the following compressed sensing \cite{donoho2006compressed} problem.
Let the original 
sparse signal vector be $\bm{x}=[x_1, x_2, \ldots, x_n]^\mathrm{\mathrm{T}} \in \mathbb{R}^{n}$ and the observation matrix be $\bm{A} \in \mathbb{R}^{m \times n} (m < n)$.
The observation model is expressed as $\bm{y} = \bm{Ax} + \bm{w}$, where the noise vector 
$\bm{w} \in \mathbb{R}^{m}$ follows Gaussian distribution with mean $\bm{0}$ and covariance $\sigma^{2}\bm{I}$.

To solve the sparse signal estimate problem, one can consider the convex optimization problem called LASSO (Least Absolute 
Shrinkage and Selection Operator)~\cite{lasso}, which corresponds to a regularized least squares method and the 
estimated vector is given by
\begin{equation}
  % \hat{\bm{x}} = \underset{\bm{x} \in \mathbb{R}^n}{\arg\min}
  \hat{\bm{x}} = \arg\min_{\bm{x} \in \mathbb{R}^n}
  \left( \frac{1}{2} ||\bm{y} - \bm{Ax}||^{2}_{2} + \tau||\bm{x}||_{1} \right),
  \label{lasso}
\end{equation}
where $||\bm{x}||_{1} = \sum^{n}_{i=1} |x_{i}|$ is the $\ell_{1}$ norm. The regularization coefficient
$\tau \in \mathbb{R} \ (\tau > 0)$ can be used to control the strength of the $\ell_1$ regularization term.
% The first term works to reduce the squared error between the original signal vector and the estimated signal vector, and 
% the second term works to make the estimated signal vector more sparse.

As a method for solving LASSO, ISTA \cite{charnbolle1998nonlinear,daubechies2004iterative} is a well-known algorithm.
ISTA is a method that iteratively minimizes the cost function of (\ref{lasso})
by the proximal gradient method \cite{parikh2014proximal}.
The iterative formula of ISTA is defined as
 \begin{align}
    \bm{z}^{(t)} &=\bm{x}^{(t)} - \eta\bm{A}^{\mathrm{T}} (\bm{Ax}^{(t)}-\bm{y}), \label{eq:ogu}\\
    \bm{x}^{(t+1)}&= {S}_{\lambda}(\bm{z}^{(t)}), \quad (t = 1, 2, \ldots, T), \label{eq:slambda}
  \end{align}
where the parameter $\eta$ is the step size,
$\lambda$ is the thresholding parameter and set to $\lambda = \tau \eta$,
% The iteration index $t$ runs from $1$ to $T$,
% where 
$T$ denotes the maximum number of iterations,
and $\bm{x}^{(1)} = \bm{0}$.
Moreover, ${S}_{\lambda}(\bm{z}^{(t)})$ is called the soft-thresholding function and is expressed as
\begin{align}
    {S}_{\lambda}(x)
    &=\begin{cases}
    x - \lambda & (x \geq \lambda) \\
    0          & (-\lambda < x < \lambda) \\
    x + \lambda & (x\leq -\lambda).
  \end{cases}
\end{align}
The soft-thresholding function ${S}_{\lambda}(\bm{z}^{(t)})$ is applied element-wise when operating on vectors.

It is known that setting the step size as $\eta = 1/L$ is optimal in terms of convergence speed, where $L$ is the maximum eigenvalue of the matrix $\bm{A}^{\mathrm{T}}\bm{A}$, that is, $L = \lambda_{\max}(\bm{A}^{\mathrm{T}}\bm{A})$.
The convergence rate of ISTA is shown to be sublinear 
and improved version called FISTA (Fast ISTA) has also been proposed \cite{beck2009fast}.

\subsection{Deep Unfolding}
% In this section, we describe deep unfolding~\cite{deepmodel} used in this paper.
Deep unfolding \cite{deepmodel} is a method in which learnable parameters are embedded into iterative algorithms and optimized by learning with data.
This approach allows the algorithm to adapt to data while preserving the structure of the original iterations, 
which leads to the reduction of the number of parameters to be optimized compared to typical deep learning models.  
% Moreover, it has been reported that deep unfolding can improve convergence speed in many applications~\cite{deng2014tutorial}.

Let $f_{\bm{\theta}}(\cdot)$ denote the output of an algorithm. 
Here, $\bm{\theta}$
% \equiv (\theta_{1}, \theta_{2}, \ldots, \theta_{m})$ 
represents the set of learnable parameters 
embedded in each iteration of the algorithm, 
e.g., step size, thresholding parameter, etc.
To train these parameters, a dataset $\mathcal{D} \equiv \{(\bm{X}_i, \bm{Y}_i)\}_{i=1}^{D}$
is prepared.
The learnable parameters $\bm{\theta}$ are adjusted through the training process of deep learning such as stochastic gradient descent and mini-batch learning
so as to minimize a loss function between the original signal $\bm{X}_{i}$ and the estimate vector $f_{\bm{\theta}}(\bm{Y}_{i})$.
% The loss function must be set to an appropriate function depending on the problem handled.
For many regression problems, the squared error function is used as the loss function.
That is given by $L(\bm{\theta}) \equiv \sum_{(\bm{X}, \bm{Y})\in \mathcal{D} } ||f_{\bm{\theta}}(\bm{Y})-\bm{X}||^{2}$. 
% Learned parameters $\bm{\theta}^\ast$ are obtained through typical training processes of deep learning such as mini-batch learning.

For stable training, 
one can use incremental learning \cite{ito2019trainable} 
that continues the learning process while increasing the number of iterations of the algorithm being trained (i.e., the number of layers in the network) 
by one each time.

\subsection{Sketching}
\label{sec:sketched-gradient}
Linear sketching \cite{woodruff2014sketching} is a dimensionality reduction technique that applies a matrix with low dimension from the left.
Many sketching literature focuses on over-constrained systems ($m>n$) rather than under-determined systems.
Typical problem settings are as follows. 
The objective is to find a vector $\bm{x} \in \mathbb{R}^n$ that minimizes the relaxed least squares error 
that satisfies $\|\bm{Ax}-\bm{y}\|_2 \leq (1\pm \epsilon)\|\bm{Ax}^\star-\bm{y}\|_2$ with probability $1-\delta$, 
where $\bm{x}^\star$ is the optimal solution of the original problem.
Sketching is a technique that replaces the problem with $\min_{\bm{x} \in \mathbb{R}^n} \| \bm{S} \bm{A} \bm{x} - \bm{S} \bm{y} \|_2^2$, 
where $\bm{S}\in \mathbb{R}^{l \times m}$ ($l \ll m$) is a sampled random matrix called sketching matrix.
For sketch matrices satisfying certain conditions, 
it follows from theory of random projection \cite{Sarlos2006-mk} 
that the relaxed problem condition is satisfied if the sketch size $l$ has the order of $\Theta(n/\epsilon^2)$.

For the sketching matrix $\bm{S}$, Gaussian random matrices, 
% partial Hadamard matrices,
subsampled randomized trigonometric transform sketching matrices, 
Count Sketch, etc., are often used \cite{woodruff2014sketching}. By
selecting an appropriate $\bm{S}$, it becomes possible to perform
dimensionality reduction while preserving the original problem structure
as much as possible and reduce computational cost.
The more specific discussion is in Sect. \ref{sec:sketchingmatrix}.

In recent signal processing, balancing the high dimensionality of data and the constraints of computational resources 
has become a major issue.
Conventional iterative optimization algorithms perform gradient calculations using all data at every step, 
so when the number of rows in the observation matrix or the amount of data is large, computational cost and
memory usage become bottlenecks.
Against this background, Tang et al.~\cite{tang2017exploiting,tang2019randomized} proposed an algorithm that integrates sketching~\cite{woodruff2014sketching} 
into optimization methods based on gradient descent methods, enabling fast and structure-preserving estimation.
As a result, it is possible to efficiently obtain a solution close to that of the original high-dimensional problem 
with low computational complexity.

We briefly review the projected gradient method with sketching proposed in \cite{tang2017exploiting}.
The target optimization problem is the following least squares problem:
\begin{equation}
    \min_{\bm{x} \in \mathcal{K}} f(\bm{x}), \ f(\bm{x}) = \frac{1}{2} \| \bm{A} \bm{x} - \bm{y} \|_2^2,
    \label{eq:orig-least-squares}
\end{equation}
where 
% $\bm{A} \in \mathbb{R}^{m \times n}$ is the observation matrix, $\bm{y} \in \mathbb{R}^m$ is the observation vector, and 
$\mathcal{K} \subset \mathbb{R}^n$ is a convex set imposing structural constraints.
A projection operator is defined as
\begin{equation}
\mathcal{P}_\mathcal{K}(\bm{x}) = 
% \underset{\bm{z} \in \mathbb{R}^n}{\arg\min} 
\arg\min_{\bm{z} \in \mathbb{R}^n}
\left\{ \mathbb{I}_\mathcal{K}(\bm{z}) + \frac{1}{2} \| \bm{x} - \bm{z} \|_2^2 \right\},
\end{equation}
where $\mathbb{I}_\mathcal{K}(\bm{z})$ is the indicator function for the set $\mathcal{K}$.
By using the projection operator, the solution to~\eqref{eq:orig-least-squares} can be obtained iteratively by the update 
formula $ \bm{x}^{(t+1)} = \mathcal{P}_\mathcal{K} \left( \bm{x}^{(t)} - \eta \nabla f(\bm{x}^{(t)}) \right), (t=1, 2, \ldots, T)$, where
$ \eta $ is the step size. 
Since $ \nabla f(\bm{x}) = \bm{A}^{\mathrm{T}} ( \bm{A} \bm{x} - \bm{y} ) $, this can be rewritten as 
$ \bm{x}^{(t+1)} = \mathcal{P}_\mathcal{K} \left( \bm{x}^{(t)} - \eta \bm{A}^{\mathrm{T}} ( \bm{A} \bm{x}^{(t)} - \bm{y} ) \right) $.

By introducing sketching, the matrix-vector product operations in $\bm{A} \bm{x}$ and 
$\bm{A}^{\mathrm{T}} ( \bm{A} \bm{x} - \bm{y} )$ used in gradient calculation can be approximated in a lower-dimensional space.
A sketching matrix $\bm{S} \in \mathbb{R}^{l \times m}$ ($l \ll m$) is introduced.
The number $l$ of rows is called sketch size. We can consider an alternative least 
squares problem
\begin{equation}
    \min_{\bm{x} \in \mathcal{K}} \tilde{f}(\bm{x}), \ \tilde{f}(\bm{x}) = \frac{1}{2} \| \bm{S} \bm{A} \bm{x} - \bm{S} \bm{y} \|_2^2.
\end{equation}
Accordingly, the gradient in this case is given by 
$ \nabla \tilde{f}(\bm{x}) = \bm{A}^{\mathrm{T}} \bm{S}^{\mathrm{T}} \bm{S} ( \bm{A} \bm{x} - \bm{y} ) $, and the update formula 
is replaced by
\begin{equation}
    \bm{x}^{(t+1)} = \mathcal{P}_\mathcal{K} \left( \bm{x}^{(t)} - \eta \bm{A}^{\mathrm{T}} \bm{S}^{\mathrm{T}} (\bm{S}\bm{A} \bm{x}^{(t)} - \bm{Sy} ) \right).
    \label{eq:sketched-update}
\end{equation}

Tang et al. theoretically demonstrated that if the sketch size $l$ is appropriately selected 
and algorithms such as GPIS and IHS are used, 
a solution close to the optimal 
solution of the original problem can be obtained~\cite{woodruff2014sketching, pilanci2016iterative}.

\section{Proposed Method}
\label{sec:proposed_method}
In this paper, we introduce sketching techniques to ISTA, which is inspired by \cite{tang2017exploiting}, 
and propose a hybrid algorithm that combines original gradient 
update with sketched gradient update. Furthermore, by applying deep unfolding to this algorithm, we make the parameters in 
each iteration learnable, aiming to improve recovery accuracy while maintaining computational efficiency. 
% Below, we explain the overview of each method and the configuration of the proposed method.

\subsection{Sketched ISTA}
To improve the computational efficiency of ISTA, we apply a sketching technique to matrix $\bm{A}$ to reduce the 
computational load of iterative calculations. 
First, by sketching $\bm{A} \in \mathbb{R}^{m \times n}$ and $\bm{y} \in \mathbb{R}^{m}$ with a 
random matrix $\bm{S} \in \mathbb{R}^{l \times m}(l \ll m)$, 
pre-computation of $\bm{S} \bm{y}$ and $\bm{S} \bm{A}$ is performed.
We then consider the transformed observation model as:
\begin{equation}
\bm{S} \bm{y} = \bm{S} \bm{A} \bm{x} + \bm{S} \bm{w},
\end{equation}
% Moreover, through sketching, we
and consider the sketched problem:
\begin{equation}
  \min_{\bm{x} \in \mathbb{R}^{n}} \frac{1}{2} \| \bm{SAx} - \bm{Sy} \|_2^2 + \tau \| \bm{x} \|_1.
  \label{eq:sketched-lasso}
\end{equation}
This allows \emph{sketched ISTA} to execute the
proximal gradient iterations as conventional ISTA, but in a
smaller dimension.
The iterative formula for sketched ISTA is given by:
\begin{align}
\bm{z}^{(t)} &= \bm{x}^{(t)} - \eta  {\bm{A}}^{\mathrm{T}} {\bm{S}}^{\mathrm{T}} (\bm{S}\bm{A} \bm{x}^{(t)} - \bm{S}{\bm{y}}), \label{eq:sgu}\\
\bm{x}^{(t+1)} &= {S}_{\lambda}(\bm{z}^{(t)}), \quad (t = 1, 2, \ldots, T), 
\end{align}
where the maximum number of iterations is $T$.
The original ISTA includes the matrix-vector product $\bm{A} \bm{x}^{(t)}$ which requires $mn$ multiplications and $m(n-1)$ additions.
By introducing the sketching matrix $\bm{S}$ and performing pre-computation of the matrix product $\bm{SA}$, 
these factors for the product $(\bm{S}\bm{A}) \bm{x}^{(t)}$ are reduced to $ln$ multiplications and $l(n-1)$ additions.
% Therefore, the number of additions and multiplications required 
% in each iteration decreases from $4mn + 3n - m$ to $4ln + 3n - l$.}

% The iteration index $t$ runs from $1$ to $T$, 
% Consequently, by introducing the sketching matrix $\bm{S}$, the number of additions and multiplications required 
% in each iteration decreases from $4mn + 3n - m$ to $4ln + 3n - l$.

We adopt the Gaussian sketch~\cite{tang2017exploiting} in this paper. Specifically, each element of the sketching matrix 
$\bm{S} \in \mathbb{R}^{l \times m}$ is independently drawn from a Gaussian distribution with zero mean and variance
$1/l$, that is, $S_{ij} \sim \mathcal{N}(0, 1/l)$. 
It should be noted that the theory of the original sketching discussed in Sect. \ref{sec:sketched-gradient} cannot be directly applied to the problem in this section 
because the problem is compressed sensing, i.e., under-determined.
In other words, there is no guarantee that the solution of the sketched ISTA is close to the solution of the original problem.

\subsection{Periodic Sketched ISTA}
In the sketched ISTA, 
the algorithm introduces a different gradient ${\bm{A}}^{\mathrm{T}} {\bm{S}}^{\mathrm{T}} (\bm{S}\bm{A} \bm{x}^{(t)} - \bm{S}{\bm{y}})$ 
instead of the original gradient ${\bm{A}}^{\mathrm{T}} (\bm{A} \bm{x}^{(t)} - {\bm{y}})$ of ISTA at every step.
The accumulated difference could encourage updates in a direction different from the desired solution 
and lead to poor performance.
One possible solution is to implement a process for pivoting toward the right direction.
In this paper, we propose \emph{PSISTA} (Periodic Sketched ISTA) that periodically switches between standard ISTA and sketched ISTA.
Figure~\ref{alg} illustrates the iterative processes of the proposed algorithm.

Specifically, the iterative process of PSISTA is defined as follows for the $t$-th iteration:

\begin{align}
 \text{OGU}(\bm{x}) &= \bm{x} - \eta\bm{A}^{\mathrm{T}} (\bm{Ax} -\bm{y}), \\
 \text{SGU}(\bm{x}) &= \bm{x} - \eta\bm{A}^{\mathrm{T}} \bm{S}^{\mathrm{T}} (\bm{S} \bm{A} \bm{x} - \bm{Sy} ), \\
  \bm{z}^{(t)} &= 
  \begin{cases}
    \text{OGU}(\bm{x}^{(t)}) & \text{if } (t-1) \bmod P = 0, \\
    \text{SGU}(\bm{x}^{(t)}) & \text{otherwise},
  \end{cases} \\
  \bm{x}^{(t+1)} &= {S}_{\lambda}(\bm{z}^{(t)}),
\end{align}
where OGU and SGU stand for Original Gradient Update \eqref{eq:ogu} and Sketched Gradient Update \eqref{eq:sgu}, respectively.

A positive integer parameter $P$ is named period and used to control the period of introducing the original gradient update. Thus, it 
becomes possible to leverage the advantages of the lower computational cost of sketched gradient update while suppressing 
the degradation of estimation accuracy by periodically inserting original gradient update. 
When $P=1$, the proposed method reduces to ISTA, since the original gradient update is used at every iteration.
Note that the original gradient update is applied at the beginning of each period in PSISTA.

The purpose of this method is to balance computational efficiency and estimation accuracy by integrating the low computational
cost of the sketched ISTA with the high representational capability of standard ISTA.
The introduction of the original gradient plays a role in periodically adjusting the direction of updates.
The period $P$ controls the trade-off between reducing computational complexity through sketching and improving accuracy through gradient adjustment.
\begin{figure}
\centering
\includegraphics[width=1.0\linewidth]{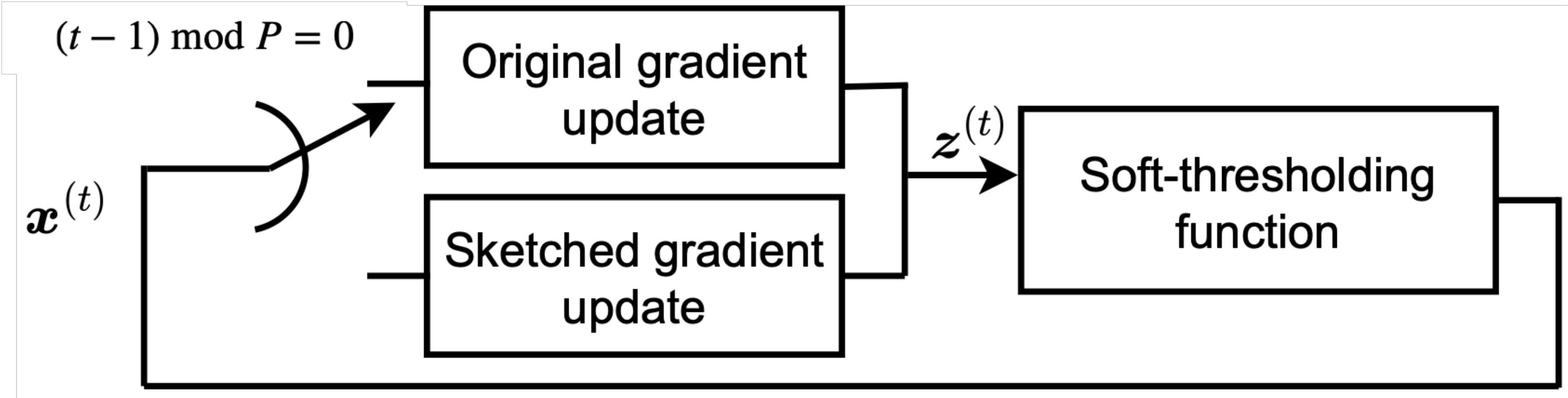}
\caption{Algorithm diagram of Periodic Sketched ISTA}
\label{alg}
\end{figure}

\subsection{PSISTA with Deep Unfolding: DU-PSISTA}
The PSISTA has the potential to flexibly control a trade-off between computational efficiency and recovery accuracy 
but there remains concerns that 
there is no performance guarantee for the sketched problem in under-determined systems 
and then the performance will monotonically decrease with the number of times sketching is introduced (i.e., $P$ increases).
% is concern that the performance will monotonically decrease with the number of times sketching is introduced (i.e., $P$ increases).
To further improve the performance of PSISTA, we apply deep unfolding to the proposed PSISTA method through data-driven hyperparameter 
learning. By unfolding each iterative step of PSISTA in the time direction, we extend it to a learnable structure. 
In each layer, parameters such as step size are introduced as learnable parameters.

We set the step size $\eta$ and thresholding parameter $\lambda$ in 
PSISTA as the learnable parameters. 
When the maximum number of iterations of the algorithm is $T$, each parameter at the $t$-th iteration can be 
expressed as $\eta_t, \lambda_t$. We learn $\bm{\eta} \equiv(\eta_1,\eta_2,\ldots,\eta_T)$ and 
$\bm{\lambda} \equiv(\lambda_1,\lambda_2,\ldots,\lambda_T)$ using deep unfolding. The learned parameters 
are used during performance measurement.
We name the proposed PSISTA to which deep unfolding is applied as 
\emph{DU-PSISTA} (Deep Unfolded-Periodic Sketched ISTA).
If these parameters are trained properly, 
it is expected to design an algorithm that achieves performance close to that of the original ISTA before dimensionality reduction by sketching, 
while maintaining high computational efficiency.

The update rules of DU-PSISTA is summarized in Algorithm~\ref{alg:du-psista}.
It should be emphasized that the sketching matrix $\bm{S}$ is fixed during the execution 
so that the pre-computation of $\bm{SA}$ and $\bm{Sy}$ is only required once before the algorithm is executed.
Furthermore, as long as the observation matrix $\bm{A}$ is fixed, 
even when new observations $\bm{y}$ are obtained, 
the pre-computed $\bm{SA}$ can be directly utilized in the algorithm.
\begin{algorithm}[t]
  \caption{DU-PSISTA}
  \label{alg:du-psista}
  \begin{algorithmic}
    \REQUIRE $\bm{x}^{(1)}=\bm{0},\bm{SA}, \bm{Sy}, P, T$, learned parameters $\bm{\eta}, \bm{\lambda}$
    \ENSURE $\bm{x}^{(T+1)}$
    \FOR{$t = 1$ to $T$}
      \IF{$(t-1) \bmod P = 0$}
        \STATE $\bm{x}^{(t+1)} = \bm{x}^{(t)} - \eta_t \bm{A}^{\mathrm{T}} (\bm{A} \bm{x}^{(t)} - \bm{y})$
      \ELSE
        \STATE $\bm{x}^{(t+1)} = \bm{x}^{(t)} - \eta_t (\bm{SA})^{\mathrm{T}} ((\bm{SA}) \bm{x}^{(t)} - \bm{Sy})$
      \ENDIF
      \STATE $\bm{x}^{(t+1)} = {S}_{\lambda_t}(\bm{x}^{(t+1)})$
    \ENDFOR
  \end{algorithmic}
\end{algorithm}

% Through deep unfolding, it is expected that the overall convergence speed and accuracy of the algorithm can be improved by 
% optimizing parameters for each iteration to data while achieving both the computational efficiency of Sketched ISTA 
% and the estimation accuracy of standard ISTA. 

\subsection{Computational Complexity}
This section evaluates the computational complexity for the standard ISTA, sketched ISTA, and 
PSISTA. The computational complexity here denotes the total number of additions and 
multiplications required for $T$ iterations of each algorithm.
% until the predicted value is output when each algorithm iterates $T$ times. 
Note that we do not include some pre-computations in the computational complexity; 
the derivation of the sketched matrix $\bm{SA}$ and vector $\bm{Sy}$, and deep unfolding 
because these are computed only once before the algorithm is executed and fixed during the execution.
% However, in SGU, ${\bm{A}}^{\top}{\bm{S}}^{\top}$, $\bm{S}\bm{A}$, and $\bm{S}{\bm{y}}$ are assumed to be pre-computed 
% and are not included in the computational complexity.  Additionally, note that the computational complexity does not 
% include the additional cost introduced by deep unfolding. 

We first consider the computational complexity per iteration for both ISTA and 
sketched ISTA.
In both cases, matrix-vector products included in gradient computation 
account for the largest proportion of addition and multiplication counts. 
Specifically, 
for the standard ISTA, $\bm{Ax}^{(t)}$ is the product of $m\times n$ matrix and $n\times 1$ vector 
so that it requires $mn$ multiplications and $m(n-1)$ additions.
For sketched ISTA, these can be reduced to $ln$ multiplications and $l(n-1)$ additions.
% matrix-vector 
% products with $\bm{A}$ or $\bm{S}\bm{A}$ are required during gradient computation, and this computation accounts 
% for the largest proportion of both addition and multiplication counts.}
% 
% In {Sketched ISTA}, since the matrix-vector multiplications in the gradient update dominate 
% the per-iteration computational complexity, the original matrix $\bm{A}$ is multiplied by a sketching 
% matrix $\bm{S}$ to reduce its size from $m \times n$ to $l \times n$. 
{By designing the sketching 
such that $l \ll m$, it becomes possible to significantly reduce the computational complexity per iteration overall computational.}
The complexity of the soft-thresholding operation can be $n$ additions.
In summary, 
{per-iteration computational complexities of ISTA}, $O_{\text{ISTA}}$, and {sketched ISTA}, $O_{\text{Sketch}}$, are given by:
\begin{align}
O_{\text{ISTA}} &= 4mn + 3n - m, \label{eq:comp_ista}\quad \\
O_{\text{Sketch}} &= 4ln + 3n - l. \quad 
\end{align}
The total complexity of $T$ iterations of the standard ISTA is $C_{\text{ISTA}} = T \cdot O_{\text{ISTA}}$.

We next consider the complexity of the proposed PSISTA, 
which includes the original ISTA-based update and sketched ISTA-based update.
The number of executions of the original ISTA-based update, $N_{\text{ISTA}}$, and sketched ISTA-based update, $N_{\text{Sketch}}$, are given by:
\begin{align}
N_{\text{ISTA}} &= \left\lfloor \frac{{T}}{P} \right\rfloor + 
\begin{cases}
% 1 & \text{if } T \bmod P \geq 1 \\
% 0 & \text{otherwise},
0 & \text{if } (T-1) \bmod P = P-1 \\
1 & \text{otherwise},
\end{cases} \\
N_{\text{Sketch}} &= T - N_{\text{ISTA}}.
\end{align}
From the above, the total computational complexity $C_{\text{PSISTA}}$ of PSISTA is:
\begin{equation}
\begin{aligned}
C_{\text{PSISTA}} &= N_{\text{ISTA}} \cdot O_{\text{ISTA}} + N_{\text{Sketch}} \cdot O_{\text{Sketch}} \\
&= N_{\text{ISTA}}(4mn + 3n - m) \!
 +\!(T - N_{\text{ISTA}})(4ln + 3n - l).
 \label{eq:comp}
\end{aligned}
\end{equation}

Introducing sketching even once during the iterative process reduces the overall computational complexity compared to standard ISTA.

\subsection{Choice of Sketching Matrix}
\label{sec:sketchingmatrix}
The computational complexity discussed in the previous subsection focuses on the matrix vector product in the algorithm update.
However, depending on a system, 
pre-computation of the sketched matrix $\bm{SA}$ may also need to be considered,
for example, when the observation matrix $\bm{A}$ changes frequently during algorithm execution.
The computational complexity of the matrix product of $\bm{SA}$ is $\mathcal{O}(lmn)$ due to the $lmn$ multiplications 
as long as the dense sketching matrix $\bm{S}$ is used such as Gaussian sketch.
If the number $T$ of algorithm iterations is as large as $l$ or $m$ and the observation matrix $\bm{A}$ switches with each algorithm execution, 
this pre-computation also dominates the computational complexity to the same order as the complexity for the algorithm update.

For such cases, 
one possible solution is to use a sparse matrix as the sketching matrix $\bm{S}$ to reduce the computational complexity of the matrix product of $\bm{SA}$.
Count Sketch~\cite{woodruff2014sketching} is the representative of the sparse sketching matrices.
It is a matrix with only one non-zero elements in each column.
The non-zero elements are $+1$ or $-1$ with equal probability.
An example of Count Sketch when $(l,m)=(3,5)$ is 
\[
  \bm{S} = \begin{pmatrix}
  0 & 1 & 1 & 0 & -1 \\
  0 & 0 & 0 & -1 & 0 \\
  1 & 0 & 0 & 0 & 0 \\
  \end{pmatrix}.
\]

The computational complexity of the matrix product of $\bm{SA}$ with Count Sketch is $\mathcal{O}(\mathrm{nnz}(\bm{A}))\leq\mathcal{O}(mn)$, 
where $\mathrm{nnz}(\bm{A})$ is the number of non-zero elements in $\bm{A}$.
This is lower than the computational complexity $\mathcal{O}(lmn)$ of that with Gaussian sketch.
However, in general, the performance of Count Sketch is not as good as that of Gaussian sketch in terms of recovery accuracy 
so that the choice of the sketching matrix depends on the system and the application.
We evaluate the performance of Count Sketch for DU-PSISTA in the computer experiments in Sect.~\ref{sec:exp_countsketch}.

\section{Convergence Analysis}
\label{sec:convergence_analysis}
In this section, we analyze the convergence of the proposed DU-PSISTA.
We introduce some assumptions and then derive the estimated error bound of the proposed DU-PSISTA.
The detailed proof of Theorem~\ref{thm:error_bound} is provided in Appendix.

We use the notation $\bm{y}=\bm{Ax}^\star + \bm{w}$ where $\bm{x}^\star$ is the true sparse signal to be estimated 
(not necessarily the optimal solution of the LASSO problem) 
with the support set $\Omega =\{i: x_i^\star \neq 0\}$.
The update rule of the proposed DU-PSISTA is summarized as follows:
\begin{align}
  \bm{x}^{(t+1)}&=S_{\lambda_t} \left(\bm{x}^{(t)}-\eta_t\bm{A}_t^\mathrm{T}\left(\bm{A}_t\bm{x}^{(t)}-\bm{y}_t\right)\right) \\ 
  &= S_{\lambda_t} \left(\bm{x}^{(t)}-\eta_t\bm{A}_t^\mathrm{T}\bm{A}_t\bm{x}^{(t)}+\eta_t\bm{A}_t^\mathrm{T}\bm{A}_t\bm{x}^\star+\eta_t\tilde{\bm{A}}_t^\mathrm{T}\bm{w}\right),
  \label{eq:update}
\end{align}
where 
\begin{equation}
  \bm{A}_t, \tilde{\bm{A}}_t, \bm{y}_t = 
  \begin{cases}        \bm{A}, \bm{A}, \bm{y} & \mathrm{if} \ (t-1)\bmod P=0, \\        \bm{SA}, \bm{S}^\mathrm{T}\bm{SA}, \bm{Sy} & \mathrm{if} \ \mathrm{otherwise}.    \end{cases}
\end{equation}

The following definition and assumptions are used to derive the estimated error bound of the proposed DU-PSISTA.
\begin{definition}
  For the matrices $\bm{A}_t=[\bm{A}_{t,1},\ldots,\bm{A}_{t,n}]=[A_{t,ij}]$, $\bm{A}=[A_{ij}]$, and $\tilde{\bm{A}}_t=[\tilde{\bm{A}}_{t,1},\ldots,\tilde{\bm{A}}_{t,n}]=[\tilde{A}_{t,ij}]$,
  \begin{align}        \max_{i,j}|\bm{A}_{t,i}^\mathrm{T}\bm{A}_{t,j}| &= \begin{cases}        \tilde \mu & \mathrm{if} \ (t-1)\bmod P=0 \\        \tilde \xi & \mathrm{if} \ \mathrm{otherwise}    \end{cases} \\         \max_{i,j}|\tilde{A}_{t,ij}| &= \begin{cases} C & \mathrm{if} \ (t-1)\bmod P=0 \\ D & \mathrm{if} \ \mathrm{otherwise}    \end{cases}   \end{align}
  \label{def:max}
\end{definition}
\begin{assumption}
  The true signal $\bm{x}^\star$ and noise vector $\bm{w}$ satisfy the following conditions:
  % \begin{align}
  %   (\bm{x}^\star,\bm{w})&\in\mathcal{X}(B,s,\sigma) \nonumber \\
  %   &:=\{(\bm{x},\bm{\omega})         | \|x_i^\star|\leq B, \forall i, \|\bm{x}^\star\|_0\leq s, \|\bm{\omega}\|_1\leq\sigma\}.
  % \end{align}
  % \begin{align}
  %   (\bm{x}^\star,\bm{w})&\in\mathcal{X}(s,\epsilon_w) \nonumber \\
  %   &:=\{(\bm{x},\bm{\omega})         | \|\bm{x}^\star\|_0\leq s, \|\bm{\omega}\|_1\leq\epsilon_w\}.
  % \end{align}
  \begin{equation}
    (\bm{x}^\star,\bm{w})\in\mathcal{X}(s,\epsilon_w)=\{(\bm{x},\bm{\omega}) | \|\bm{x}^\star\|_0\leq s, \|\bm{\omega}\|_1\leq\epsilon_w\}.
  \end{equation}
  \label{assum:xw}
\end{assumption}
\begin{assumption}
  % The sketching matrix $\bm{S}$ is Gaussian sketch,
  % where each element is drawn from the Gaussian distribution with mean 0 and variance $1/l$.
  Each matrix $\bm{A}_t\in\{\bm{A}, \bm{SA}\}$ satisfies RIP (Restricted Isometry Property) \cite{candes2005decoding} of order $2s$ with RIP constant $\delta_{2s}\in(0,1)$.
  Specifically, for any $2s$-sparse vector $\bm{v}\in\mathbb{R}^n$, 
  $(1-\delta_{2s})\|\bm{v}\|^2 \leq \|\bm{A}_t\bm{v}\|^2 \leq (1+\delta_{2s})\|\bm{v}\|^2$ holds.
  Additionally, the sketching matrix $\bm{S}$ is a random projection matrix that satisfies JL (Johnson-Lindenstrauss) lemma \cite{johnson1984extensions}.
  \label{assum:sketch}
\end{assumption}
\begin{assumption}
  For $\forall i\notin \Omega $, the following conditions hold:
  \begin{align}
    \lambda_t\geq\begin{cases}            \eta_t(\tilde{\mu}\|\bm{x}^{(t)}-\bm{x}^\star\|_1+C\epsilon_w) & \mathrm{if} \ (t-1)\bmod P=0 \\            \eta_t(\tilde{\xi}\|\bm{x}^{(t)}-\bm{x}^\star\|_1+D\epsilon_w) & \mathrm{if} \ \mathrm{otherwise}    \end{cases}
  \end{align}
  \label{assum:lambda}
\end{assumption}
\begin{assumption}
  The learned parameters of the proposed DU-PSISTA satisfy the following conditions using constants $\tilde{\eta}_C, \tilde{\eta}_D, \tilde{\lambda}>0$:
  \begin{align}        \sum_{t\in[T],(t-1)\bmod{P}=0}\eta_t &\leq \tilde{\eta}_C, \\         \sum_{t\in[T],(t-1)\bmod{P}\neq0}\eta_t &\leq \tilde{\eta}_D, \\         \sum_{t=1}^T \lambda_t &\leq \tilde{\lambda}.    \end{align}
  \label{assum:eta}
\end{assumption}
% \begin{assumption}
%   The learned parameters of the proposed DU-PSISTA satisfy the following condition:
%   \begin{align}
%     \|\bm{I}-\eta_t\bm{A}_t^\mathrm{T}\bm{A}_t\|<1 \ \forall t.
%   \end{align}
%   \label{assum:eye}
% \end{assumption}
Note that the deterministically bounded noise condition in Assumption~\ref{assum:xw} does not correspond to the observation model with Gaussian noise 
and the parameter $\lambda_t$ in Assumption~\ref{assum:lambda} is determined through learning in practice, 
but they are basic assumptions that are also used in \cite{chen2018theoretical} and sufficient conditions used only for analysis, not for algorithm design.

The following theorem is the main result of this section.
\begin{theorem}
  On Assumptions \ref{assum:xw}--\ref{assum:eta}, the estimated error bound of the proposed DU-PSISTA is given by:
  \begin{align}
    \|\bm{x}^{(t+1)}-\bm{x}^\star\|<&\left(\prod_{t'=1}^t\|\bm{I}-\eta_{t'}\bm{A}_{t'}^\mathrm{T}\bm{A}_{t'}\|\right)\|\bm{x}^{(1)}-\bm{x}^\star\| \nonumber \\
    &+\sqrt{s}\epsilon_w(\tilde{\eta}_C C+\tilde{\eta}_D D)+\sqrt{s}\tilde{\lambda}.
    \label{eq:theorem}
  \end{align}
  \label{thm:error_bound}
\end{theorem}

From the theorem, 
the following colloraries are derived.
\begin{corollary}
  The estimated error bound of the proposed DU-PSISTA converges linearly to a constant value depending on the sparsity $s$ and the noise level $\epsilon_w$ 
  with the properly selected step size $\eta_t$.
\end{corollary}
\begin{proof}
  The estimated error bound of the proposed DU-PSISTA can be written as in \eqref{eq:h_bound_t} and 
  the error vector $\bm{x}^{(t)}-\bm{x}^\star$ is at most $2s$-sparse.
  The error dynamics are governed by the linear operator $\bm{W}_t = \bm{I} - \eta_t \bm{A}_t^\mathrm{T}\bm{A}_t$.
  By Assumption~\ref{assum:sketch} and setting $0<\eta_t<2/\sigma_{\max}^2(\bm{A}_t)$, where $\sigma_{\max}^2(\bm{A}_t)$ is the largest singular value of $\bm{A}_t$, 
  the spectral radius of $\bm{W}_t$ is restricted to be less than 1.
  This ensures that the mapping in \eqref{eq:h_bound_t} is a contraction.
  It should be noted that, for the matrix $\bm{A}_t=\bm{A}$, the discussion reduces to that of standard compressed sensing \cite{candes2005decoding}, 
  and for the matrix $\bm{A}_t=\bm{SA}$, the RIP is preserved with high probability when the sketching matrix $\bm{S}$ satisfies JL lemma \cite{baraniuk2008simple,krahmer2011new}.
  Gaussian sketch is the typical random projection that satisfies JL lemma.
\end{proof}
\begin{corollary}
  The estimated error introduced by linear sketching in DU-PSISTA has possibility to be compensated for by adjusting the learnable parameters.
\end{corollary}
\begin{proof}
  Let 
  \begin{equation}
    e=\sqrt{s}\epsilon_w(\tilde{\eta}_C C+\tilde{\eta}_D D)+\sqrt{s}\tilde{\lambda}
  \end{equation}
  in \eqref{eq:theorem} denotes the convergence constant of the DU-PSISTA.
  When assuming no sketch is used (named as DU-ISTA), let the learnable parameters be $\eta_t',\lambda_t'$, and set $\sum_{t=1}^T \eta_t'=\eta', \ \sum_{t=1}^T \lambda_t'=\lambda'$. 
  The convergence constant in this case is 
  \begin{equation}
    e'=\sqrt{s}\epsilon_w\eta'C+\sqrt{s}\lambda'.
  \end{equation}
  The difference in the error of DU-PSISTA from that of DU-ISTA is 
  \begin{equation}
    e-e'=\sqrt{s}\epsilon_w\left( (\tilde{\eta}_C-\eta')C+\tilde{\eta}_D D\right)+\sqrt{s}(\tilde{\lambda}-\lambda').
  \end{equation}
  and it is desired to be improved as much as possible by controlling the learnable parameters of DU-PSISTA.
  For $\lambda_t$, if the proposed DU-PSISTA can learn such that $\tilde{\lambda}\approx\lambda'$, the difference $\sqrt{s}(\tilde{\lambda}-\lambda')$ can be reduced to zero.
  Using sketching inevitably introduces a term $\tilde{\eta}_D D$, but if we can learn such that $(\tilde{\eta}_C-\eta')C\approx-\tilde{\eta}_D D$, i.e., $\tilde{\eta}_C\approx\eta'-\tilde{\eta}_D\frac{D}{C}$, 
  we can reduce the term $(\tilde{\eta}_C-\eta')C+\tilde{\eta}_D D$ to zero
  as long as Assumption~\ref{assum:eta} holds.
\end{proof}
\begin{corollary}
  When relying solely on sketching in the proposed DU-PSISTA, convergence will be delayed.
  \label{corr:delay}
\end{corollary}
\begin{proof}
  When all iterations are treated as sketches and no original gradient update is introduced into DU-PSISTA (named DU-Sketched ISTA), 
  let $\bar{\eta}=\sum_{t=1}^T \eta_t$ and $\bar{\lambda}=\sum_{t=1}^T\lambda_t$.
  The convergence constant of DU-Sketched ISTA is $e=\sqrt{s}\epsilon_w\bar{\eta}D+\sqrt{s}\bar{\lambda}$ and 
  the difference in error from the error of DU-ISTA becomes 
  \begin{equation}
    e-e'=\sqrt{s}\epsilon_w(\bar{\eta}D-\eta'C)+\sqrt{s}(\bar{\lambda}-\lambda').
  \end{equation}
  If $\bar{\eta}\approx\eta'\cdot C/D$ can be learned, the difference can be reduced to zero. 
  However, in practice, even after learning, the performance of DU-Sketched ISTA is expected to be very poor.
  % showing a significant divergence from DU-ISTA.
  The reason lies in the fact that the maximum component value of $\bm{S}^\mathrm{T}\bm{SA}$ (i.e., $D$ rather than $C$) can be larger than that of the original matrix $\bm{A}$.
  Assuming the components of matrix $\bm{A}$ follows a distribution with mean 0 and variance 1, 
  the variance of $\bm{S}^\mathrm{T}\bm{SA}$ becomes $1+m/l$ due to the effect of $\sum_{k \neq i} Q_{ik} A_{kj}$ when setting $\bm{Q}=\bm{S}^\mathrm{T}\bm{S}$. 
  Since the maximum value is proportional to the standard deviation, we can say $C \propto 1$ and $D \propto \sqrt{1 + m/l}$. Therefore,
  \begin{equation}
    \frac{C}{D}\approx\sqrt{\frac{l}{l+m}}.
    \label{eq:cd}
  \end{equation}
  For example, when $l=m/2$, we have $C/D\approx\frac{1}{\sqrt{3}}\approx0.577$. 
  That is, the total step size $\bar{\eta}\approx\eta'\frac{C}{D}$ must be smaller than that in DU-ISTA $\eta'$ by a factor of 0.5 on average. 
  Consequently, convergence slows down.
\end{proof}
\begin{corollary}
  The introduction of the original gradient updates in the proposed DU-PSISTA has ability to improve the convergence compared to that relying solely on sketching.
  \label{corr:improve}
\end{corollary}
\begin{proof}
  Even when introducing the original gradient update in DU-PSISTA, 
  the factor $C/D$ also appears in $\tilde{\eta}_D \approx \frac{C}{D}(\eta' - \tilde{\eta}_C)$.
  % the same term affects the result (where $\tilde{\eta}_D \approx \frac{C}{D}(\eta' - \tilde{\eta}_C)$). 
  However, treating $\tilde{\eta}_C$ as equivalent to $\eta'$ in DU-ISTA can reduce this effect. 
  Therefore, it is important to include the original gradient update in DU-PSISTA. 
  Furthermore, when leanable parameters in DU-PSISTA are learned properly, 
  the step sizes for iterations $(t-1)\bmod P=0$ (using original gradient update) are expected to take relatively larger values to make $\tilde{\eta}_C$ as close as possible to $\eta'$ 
  and the step sizes for iterations $(t-1)\bmod P\neq0$ (using sketched gradient update) are expected to take smaller values.
\end{proof}

\section{Numerical Experiments}
\label{sec:experiments}
\subsection{Experimental Settings}
In this section, we evaluate the performance of the proposed DU-PSISTA through numerical experiments.
% We describe the experimental settings of the computer experiments. 
% We compare the proposed DU-PSISTA with the conventional DU-ISTA.
We compare the performance of the proposed DU-PSISTA with 
the standard ISTA \eqref{eq:ogu} \eqref{eq:slambda}, DU-ISTA (with high computational complexity due to not including sketching), 
sketched ISTA (not including deep unfolding), DU-Sketched ISTA (sketched ISTA with deep unfolding).
The observation matrix $\bm{A} \in \mathbb{R}^{m \times n}$ 
was generated with elements drawn from normal distribution with mean 0 and variance 1.
The Gaussian sketch was used as the sketching matrix $\bm{S} \in \mathbb{R}^{l \times m}$ 
unless otherwise specified.
Each element of the original signal $\bm{x} \in \mathbb{R}^n$ follows i.i.d. Bernoulli-Gaussian distribution 
where it is 0 with probability 0.95 and follows $\mathcal{N}(0,1)$ with probability 0.05.
{We set the noise variance to $\sigma^2 = 0.01$}. 
The performance was evaluated with large size system $(n, m) = (1024, 512)$ and small size system $(256, 128)$.

The computational complexity of PSISTA for each system is calculated from \eqref{eq:comp} and summarized in Tables~\ref{keisanryou1} and~\ref{keisanryou2}. 
From the tables, it can be observed that the computational complexity decreases as the period $P$ increases and the sketch size $l$ decreases.
Considering both the theoretical computational complexity and the algorithmic performance evaluated later, 
we assess the performance trade-offs of the proposed method.
\begin{table}[t]
\centering
\caption{Total computational complexity $C_{\text{PSISTA}}$ for each sketch size $l$ and period $P$ for large system ($n=1024$, $m=512$, $T=40$)}
\begin{tabular}{c|rrrr}
\hline
$P \backslash l$ & 256           & 128           & 64            & 32            \\
\hline
1 $(C_{\text{ISTA}})$               & 83,988,480    & 83,988,480    & 83,988,480    & 83,988,480    \\
                 & (100.0\%)     & (100.0\%)     & (100.0\%)     & (100.0\%)     \\
2                & 63,022,080    & 52,538,880    & 47,297,280    & 44,676,480    \\
                 & (75.0\%)      & (62.6\%)      & (56.3\%)      & (53.2\%)      \\
3                & 56,732,160    & 43,104,000    & 36,289,920    & 32,882,880    \\
                 & (67.5\%)      & (51.3\%)      & (43.2\%)      & (39.2\%)      \\
5                & 50,442,240    & 33,669,120    & 25,282,560    & 21,089,280    \\
                 & (60.1\%)      & (40.1\%)      & (30.1\%)      & (25.1\%)      \\
8                & 47,297,280    & 28,951,680    & 19,778,880    & 15,192,480    \\
                 & (56.3\%)      & (34.5\%)      & (23.5\%)      & (18.1\%)      \\
\hline
\end{tabular}
\label{keisanryou1}
\end{table}
\begin{table}[t]
\centering
\caption{Total computational complexity $C_{\text{PSISTA}}$ for each sketch size $l$ and period $P$ for small system ($n=256$, $m=128$, $T=40$)}
\begin{tabular}{c|rrrr}
\hline
$P \backslash l$ & 64           & 32           & 16           & 8            \\
\hline
1 $(C_{\text{ISTA}})$                & 5,268,480    & 5,268,480    & 5,268,480    & 5,268,480    \\
                 & (100.0\%)    & (100.0\%)    & (100.0\%)    & (100.0\%)    \\
2                & 3,959,040    & 3,304,320    & 2,976,960    & 2,813,280    \\
                 & (75.1\%)     & (62.7\%)     & (56.5\%)     & (53.4\%)     \\
3                & 3,566,208    & 2,715,072    & 2,289,504    & 2,076,720    \\
                 & (67.7\%)     & (51.5\%)     & (43.5\%)     & (39.4\%)     \\
5                & 3,173,376    & 2,125,824    & 1,602,048    & 1,340,160    \\
                 & (60.2\%)     & (40.3\%)     & (30.4\%)     & (25.4\%)     \\
8                & 2,976,960    & 1,831,200    & 1,258,320    &   971,880    \\
                 & (56.5\%)     & (34.8\%)     & (23.9\%)     & (18.4\%)     \\
\hline
\end{tabular}
\label{keisanryou2}
\end{table}

The step size $\eta$ and thresholding parameter $\lambda$ {of ISTA} were both set to $1/\lambda_{\max}(\bm{A}^{\mathrm{T}}\bm{A})$. 
The initial values for the learnable parameters $\eta_t$ and $\lambda_t$ were also set to this value.
Performance was evaluated by MSE (mean squared error) over 50 different $\bm{A}$,$\bm{S}$ pairs, each with 50 samples of signals $(\bm{x}, \bm{y})$. 
% ensuring robust results across varied conditions.  
% This experimental setup was designed to verify that the proposed method achieves stable and high-precision estimation under 
% a variety of conditions, rather than depending on a single case. 
We used incremental learning~\cite{ito2019trainable} and mini-batch learning for deep unfolding. 
The mini-batch size and the number of inner loop iterations were set to 
% $B = 50$ and $I = 50$, 
$50$. 
{In each inner loop, the matrices $\bm{S}$ and $\bm{A}$ were regenerated, and mini-batches were constructed accordingly for training. }
The mean squared loss function was used. Parameter updates were performed using the Adam optimizer.

\subsection{Verification of the Effect of Period }
\begin{figure}[t]
  \centering
  \includegraphics[width=1\linewidth]{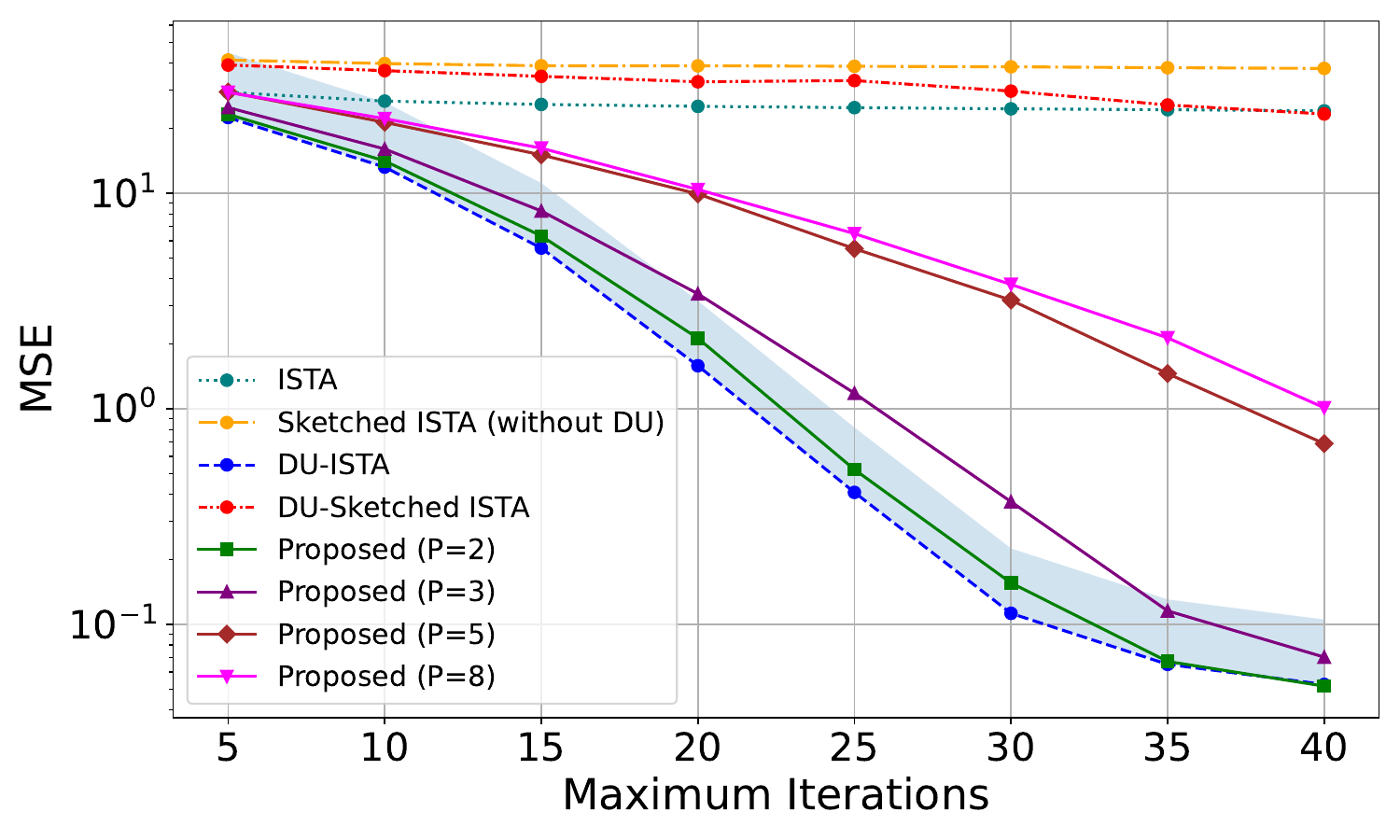}
  \caption{Comparison of MSE performance for different period $P$ in {DU-PSISTA} for large system ($n=1024, m=512, l=256$).}
  \label{jikken1}
\end{figure}
\begin{figure}[t]
  \centering
  \includegraphics[width=1\linewidth]{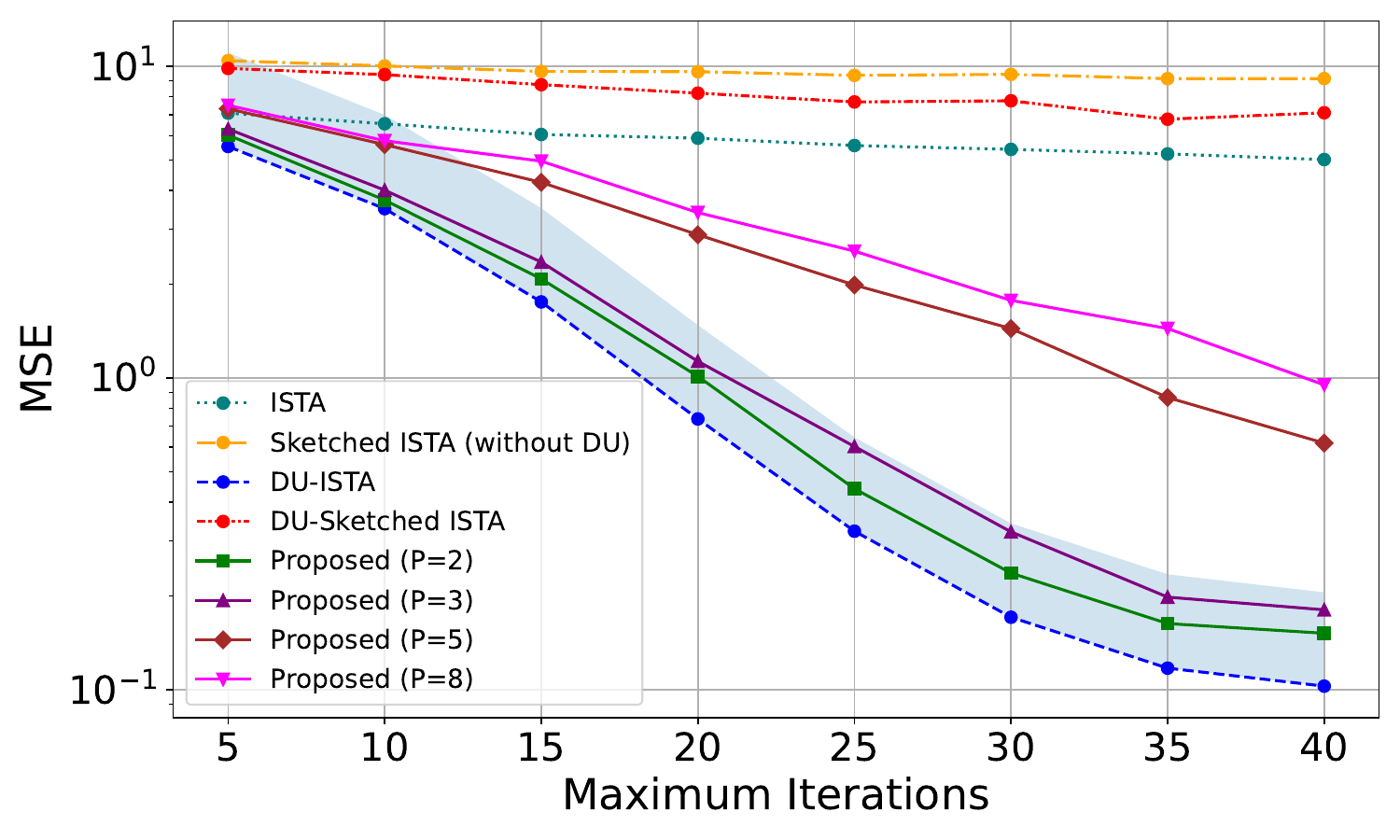}
  \caption{Comparison of MSE performance for different period $P$ in {DU-PSISTA} for small system ($n=256, m=128, l=64$).}
  \label{jikken2}
\end{figure}
In this section, we investigate the impact of varying the period $P$ in DU-PSISTA on recovery performance. 
% The experiments were conducted for two system sizes, where $(m, n, l) = (1024, 512, 256)$ and $(256, 128, 64)$.
We set $l=256, 64$ for the large and small systems, respectively.
The period $P$ was varied among $2$, $3$, $5$, and $8$ to evaluate its impact on recovery performance.

Figures~\ref{jikken1} and~\ref{jikken2} present the MSE performance of each system with varying the period $P$.
In each figure,
% the change in recovery accuracy for the larger system size. 
the blue line is the baseline DU-ISTA 
and the light blue area above the line indicates the range within {2} times the MSE of DU-ISTA.  
As shown in the figures, 
the recovery accuracy deteriorates with increasing $P$, and for larger $P$, the MSE exceeds {2} times that of DU-ISTA, 
which is considered to be due to the accumulation of approximation errors caused by more frequent use of the sketching matrix. 
On the other hand, when $P=2$, DU-PSISTA achieves accuracy comparable to DU-ISTA. 
These results indicate a trade-off between 
computational efficiency and estimation accuracy, and suggest that a smaller period $P$ is preferable for maintaining high recovery performance.
% Next, Figure~\ref{jikken2} presents the recovery accuracy for the reduced system size. Table~\ref{keisanryou2} summarizes the 
% changes in computational complexity for this case. Similarly, the performance degrades as $P$ increases, which is considered 
% to be due to the accumulation of approximation errors caused by more frequent use of the sketching matrix.

In both system sizes, DU-PSISTA demonstrated excellent performance in reducing MSE and achieving fast convergence when the 
period $P$ was set to a small value, such as 2 or 3. 
% Even in these cases, computational complexity was reduced compared to DU-ISTA,
Specifically, considering the computational complexity listed in Tables~\ref{keisanryou1} and~\ref{keisanryou2}, 
the proposed method can reduce the computational complexity to 68--75\% while keeping MSE within or near the range of ${2}$ times of that of DU-ISTA.
It indicates 
that the proposed method achieves a good balance between efficiency and accuracy.

\subsection{Verification of the Effect of {Sketch Size} }
\begin{figure}[t]
  \centering
  \includegraphics[width=1\linewidth]{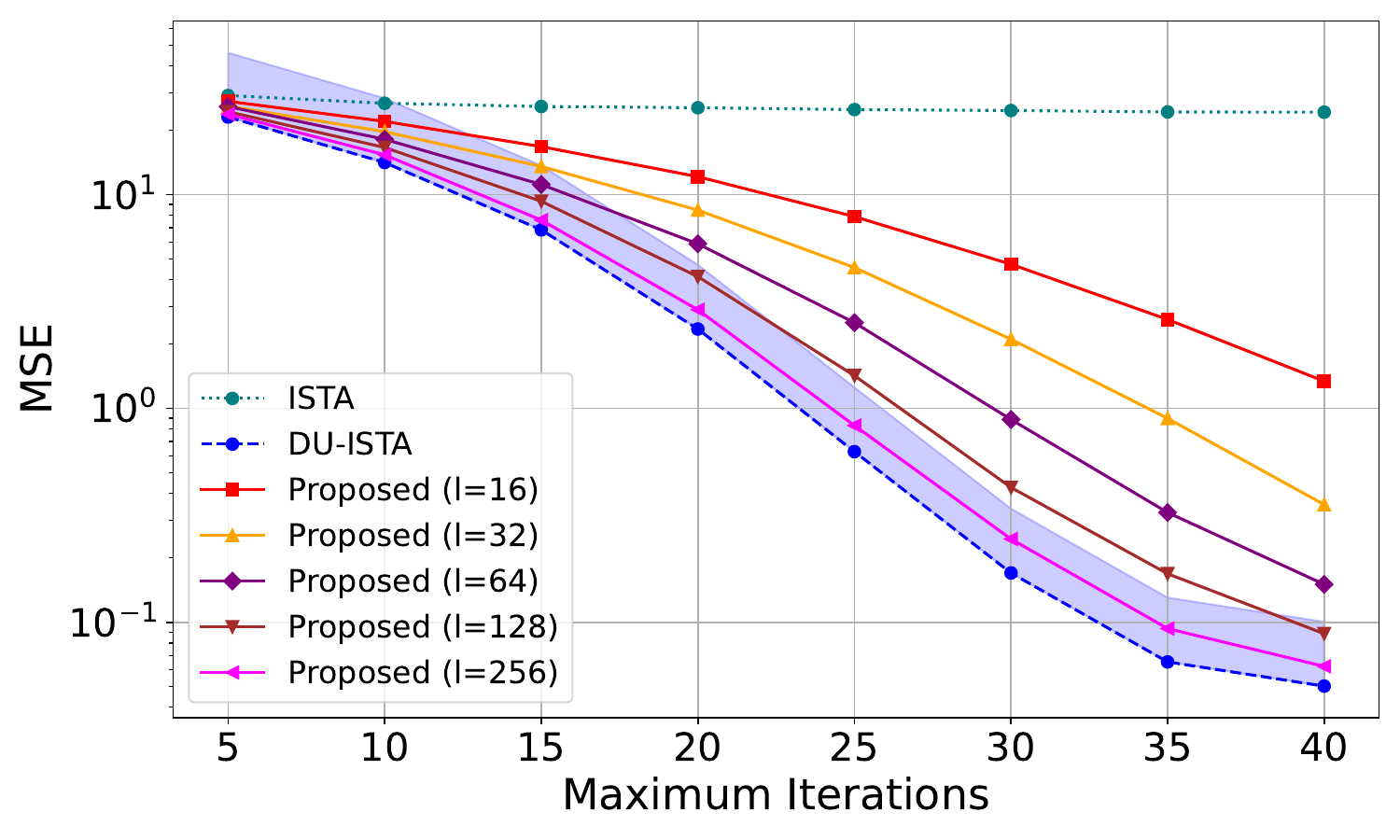}
  % \caption{Comparison of MSE performance for different numbers of rows $l$ after transformation in {DU-PSISTA} ($n=1024, m=512, P=2$)}
  \caption{Comparison of MSE performance for different sketch size $l$ in {DU-PSISTA} for large system ($n=1024, m=512, P=2$).}
  \label{jikken3}
\end{figure}
\begin{figure}[t]
  \centering
  \includegraphics[width=1\linewidth]{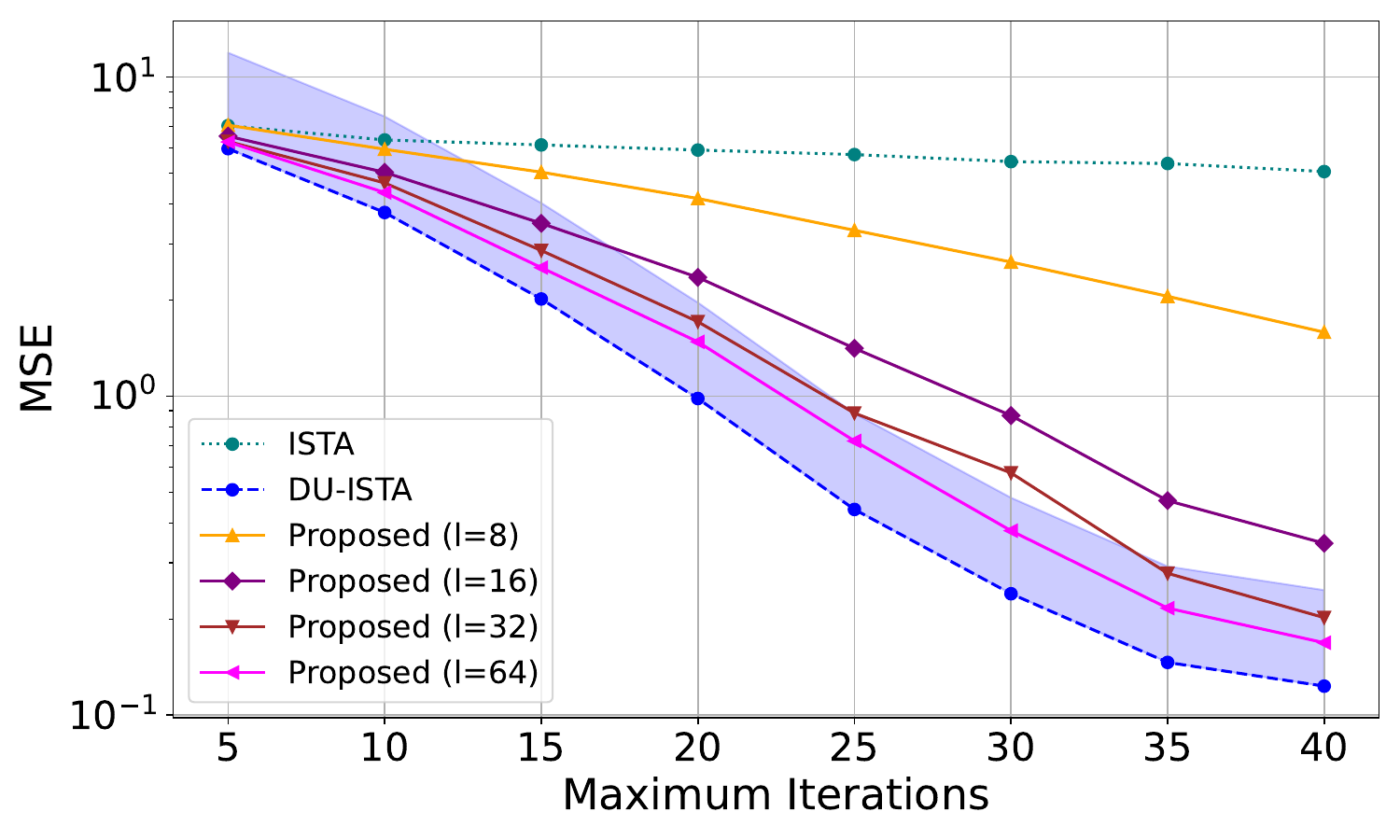}
  % \caption{Comparison of MSE performance for different numbers of rows $l$ after transformation in {DU-PSISTA} ($n=256, m=128, P=2$)}
  \caption{Comparison of MSE performance for different sketch size $l$ in {DU-PSISTA} for small system ($n=256, m=128, P=2$).}
  \label{jikken4}
\end{figure}
In this section, we investigate the impact of varying sketch size $l$ in DU-PSISTA on recovery performance. 
% The experiments were conducted for two system sizes, where $(m, n, P) = (1024, 512, 2)$ and $(256, 128, 2)$.
The sketch size $l$ was varied among $\{256, 128, 64, 32, 16\}$ for the large system and $\{64, 32, 16, 8\}$ for the small system.
We set $P=2$ for both systems.

Figures~\ref{jikken3} and~\ref{jikken4} show the recovery accuracy for the large and small systems, respectively. 
% The figures present the MSE performance of ISTA, DU-ISTA, and DU-PSISTA for each value of $l$. 
The blue line is the baseline DU-ISTA, and the blue area above the line indicates the range within {2} times the MSE of DU-ISTA. 
For the large system in Fig.~\ref{jikken3}, it can be observed that when $l = 256$, the MSE of DU-PSISTA remains within {2} times that of DU-ISTA across all iterations $T$. 
However, for $l = 128$ and smaller, the MSE increases as $l$ decreases,
indicating a degradation in recovery accuracy.
Similarly for the small system in Fig.~\ref{jikken4}, the case of $l = 64$ is within the 
range of {2} times the MSE of DU-ISTA but the other cases are not.

% Next, Figure~\ref{jikken4} shows the recovery accuracy for the smaller system size. Similarly, for $l = 64$, the MSE remains
% within 1.5 times that of DU-ISTA. However, when $l$ is reduced to $32$ or smaller, the MSE becomes noticeably higher, particularly as $T$ increases.

From these results, it is evident that reducing the sketch size $l$ leads to a decline in estimation
accuracy for both system sizes. This is likely because the reduced dimensionality of the feature space after sketching limits
the ability to sufficiently capture essential signal features. However, by varying $l$ from 256 to 128, 64, 32, and further to 16, 
the computational complexity is significantly reduced, resulting in improved computational efficiency. Therefore, 
the proposed DU-PSISTA demonstrates a good balance between accuracy and computational cost when the number of rows $ l $ is relatively large, 
such as 256 for the larger system and 64 for the smaller system.

\subsection{Execution Time}
We measured the execution time of the proposed DU-PSISTA with $(P,l)=(2,256)$ and compared it with the conventional DU-ISTA for the large system.
The execution time was measured on a machine with 
CPU: Intel Core i9-9900KF@3.60GHz, 
RAM: 16GB, and
GPU: NVIDIA GeForce RTX 2080 Ti.
As in Table \ref{keisanryou1}, 
the pre-comuputations such as the computation of $\bm{SA}$ and parameter learning in deep unfolding are not included in the execution time.

Figure \ref{executiontime} shows the execution time in seconds for the proposed DU-PSISTA and DU-ISTA 
when the maximum number $T$ of iterations was set to $20, 40, 60, 80, 100$, 
where the theoretical computational complexities per iteration are shown in \eqref{eq:comp} and \eqref{eq:comp_ista}, respectively.
The time was the average of 50 runs of each algorithm.
When $T=20$, the execution time of DU-PSISTA is about $0.08$ seconds larger than that of DU-ISTA.
One possible cause is the memory accesses incurred in the proposed method by switching between the original gradient descent and its sketch version at each iteration, 
i.e., by calling different operations.
However, for larger $T$, 
the execution time of DU-PSISTA is smaller than that of DU-ISTA.
For example, when $T=40$, the execution time of DU-PSISTA is $75\%$ of that of DU-ISTA.
This is consistent with the theoretical computational complexity shown in Table \ref{keisanryou1}.
Considering the results in Fig. \ref{jikken1}, 
the proposed method can actually achieve both high recovery accuracy and low computational cost.
\begin{figure}[t]
  \centering
  \includegraphics[width=1\linewidth]{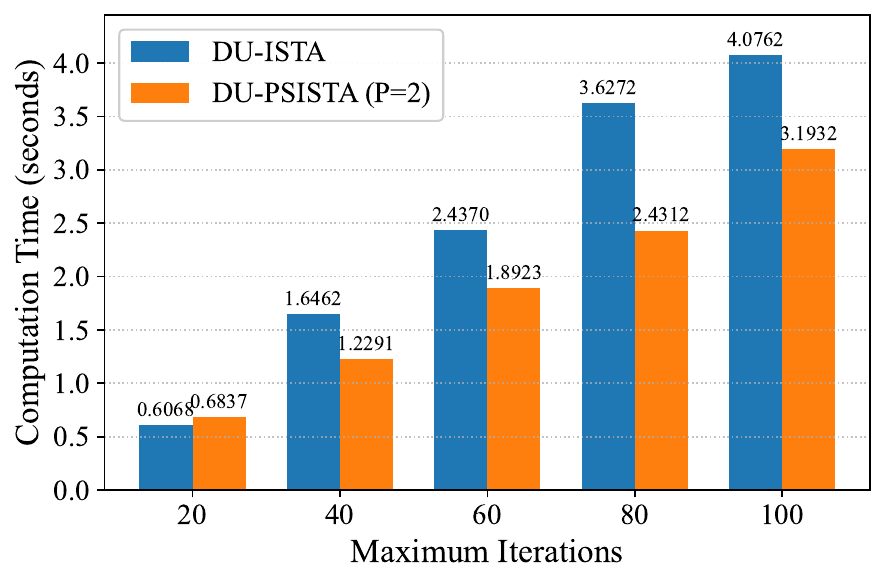}
  % \caption{Comparison of MSE performance for different numbers of rows $l$ after transformation in {DU-PSISTA} ($n=1024, m=512, P=2$)}
  \caption{Comparison of execution time for large system ($n=1024, m=512, P=2, l=256$).}
  \label{executiontime}
\end{figure}

\subsection{Performance Difference between Gaussian Sketch and Count Sketch}
\label{sec:exp_countsketch}
\begin{figure}[t]
  \centering
  \includegraphics[width=1\linewidth]{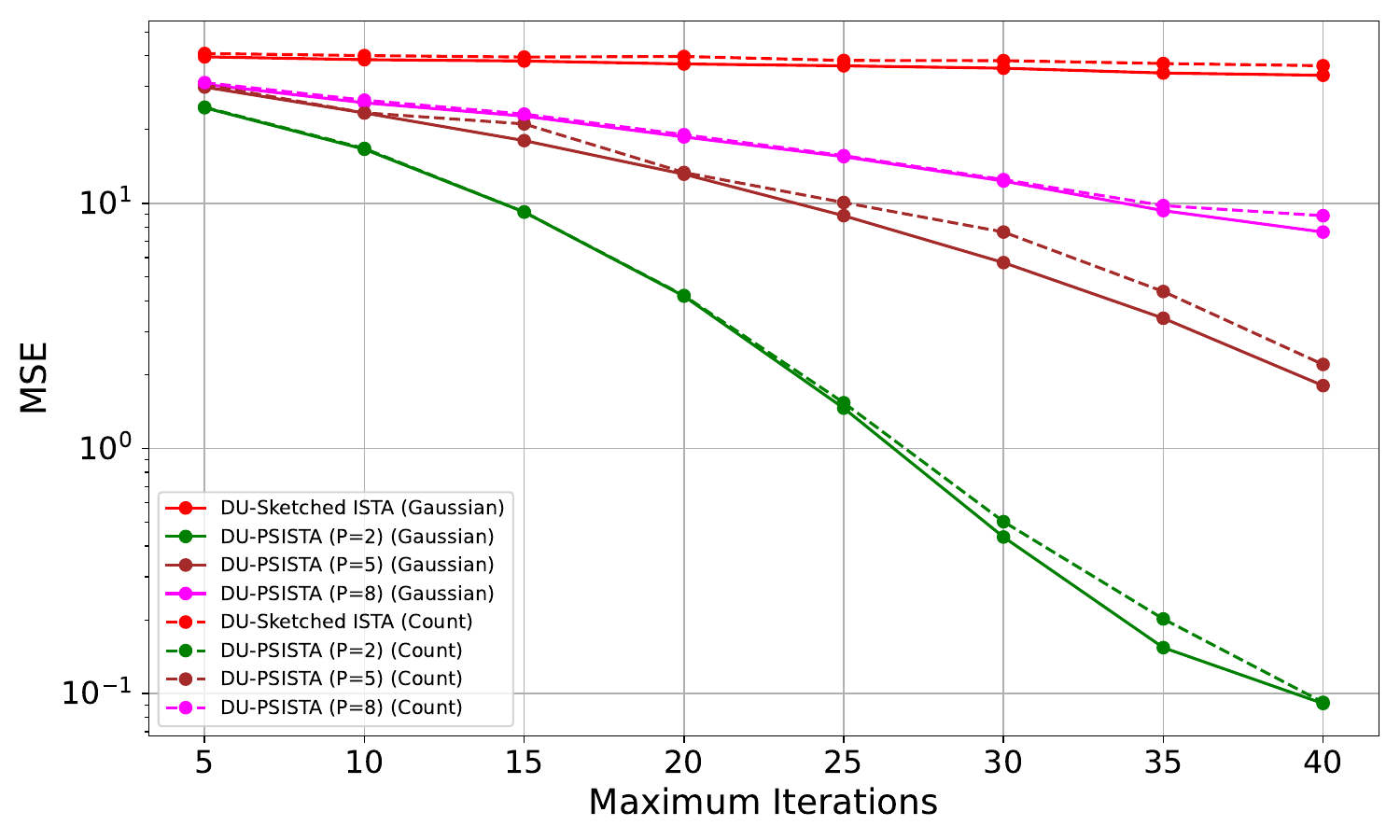}
  \caption{Comparison of MSE performance for Gaussian sketch and Count Sketch in {DU-PSISTA} for large system ($n=1024, m=512, l=256$).}
  \label{count}
\end{figure}
We investigate the impact of the choice of sketching matrix on the recovery performance. 
% The experiments were conducted for two system sizes, where $(m, n, P) = (1024, 512, 2)$ and $(256, 128, 2)$.
We used the Gaussian sketch and Count Sketch as the sketching matrices for the proposed DU-PSISTA and DU-Sketched ISTA, 
and set the large system $(n, m) = (1024, 512)$ with $l=256$.
The period $P$ was set to $2,5,8$, 
and the learning rates of deep unfolding for the algorithms with each period were set to $2.5\times10^{-5}, 10^{-5}, 4.0\times10^{-6}$, respectively.
The learning rate for DU-Sketched ISTA was set to $10^{-6}$.

Figure~\ref{count} shows the recovery accuracy for DU-Sketched ISTA and DU-PSISTA with Gaussian sketch and Count Sketch, respectively. 
In all cases, 
the recovery accuracy of algorithms using Gaussian sketch is better than that using Count Sketch 
but the improvement is not significant.
Therefore, in this system parameter setting, 
it is possible to design an algorithm that maintains equivalent performance 
while reducing the pre-computational load on the $\bm{SA}$ through Count Sketch.
This indicates that the proposed method has flexibility in the choice of sketching matrix 
according to the system and the application requirements.

% \begin{figure}[t]
%   \centering
%   \includegraphics[width=1\linewidth]{3-6_executiontime_count.pdf}
%   % \caption{Comparison of MSE performance for different numbers of rows $l$ after transformation in {DU-PSISTA} ($n=256, m=128, P=2$)}
%   \caption{Comparison of execution time for different sketching matrix in {DU-PSISTA} for small system ($n=256, m=128, P=2$).}
%   \label{executiontime_count}
% \end{figure}

% Next, Figure~\ref{jikken4} shows the recovery accuracy for the smaller system size. Similarly, for $l = 64$, the MSE remains
% within 1.5 times that of DU-ISTA. However, when $l$ is reduced to $32$ or smaller, the MSE becomes noticeably higher, particularly as $T$ increases.

\subsection{Learned Parameters}
\begin{figure}[t]
  \centering
  \includegraphics[width=1\linewidth]{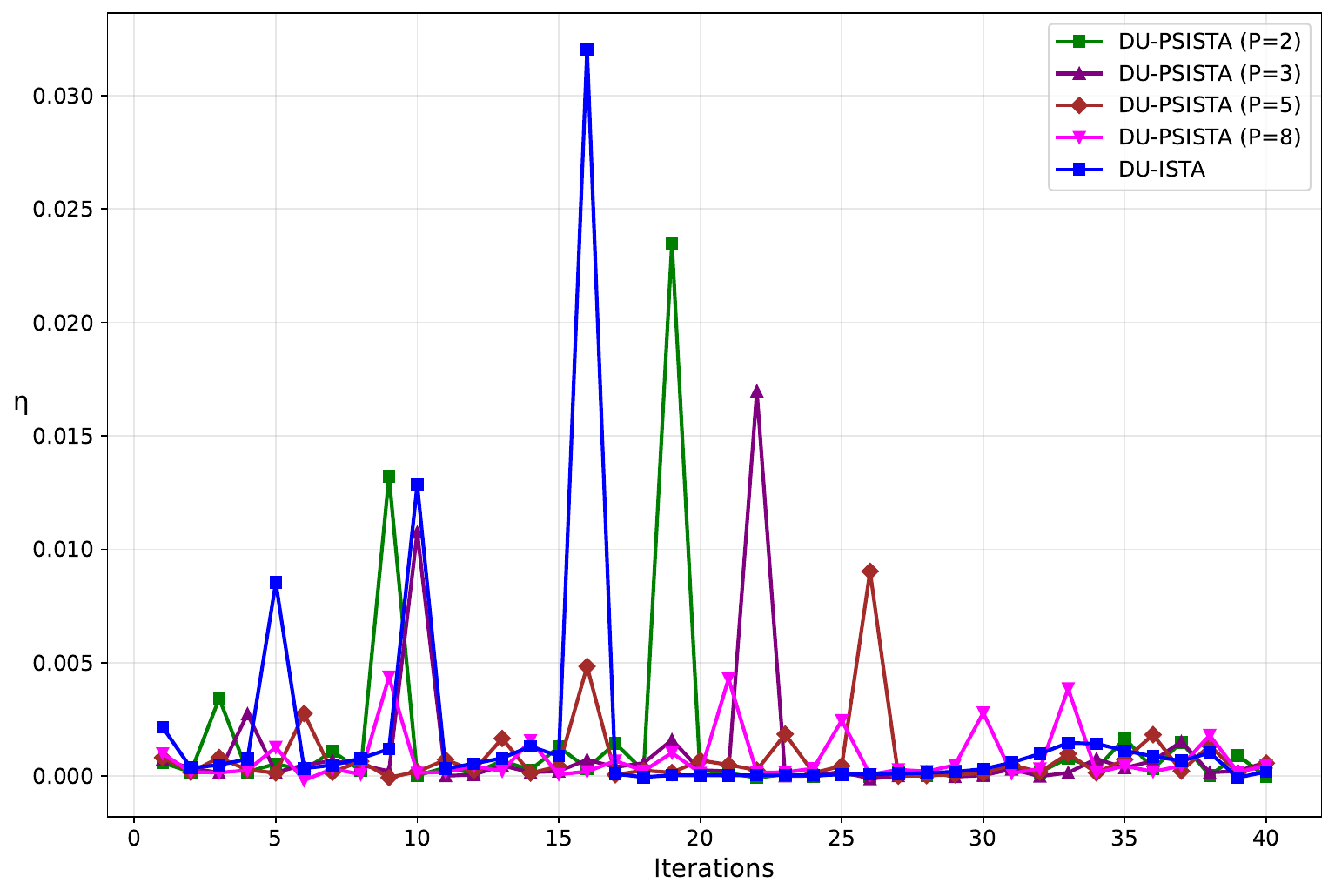}
  \caption{Learned step sizes $\eta_t$ for different period $P$ in {DU-PSISTA} for large system ($n=1024, m=512, l=256$).}
  \label{p_large_etas}
\end{figure}
\begin{figure}[t]
  \centering
  \includegraphics[width=1\linewidth]{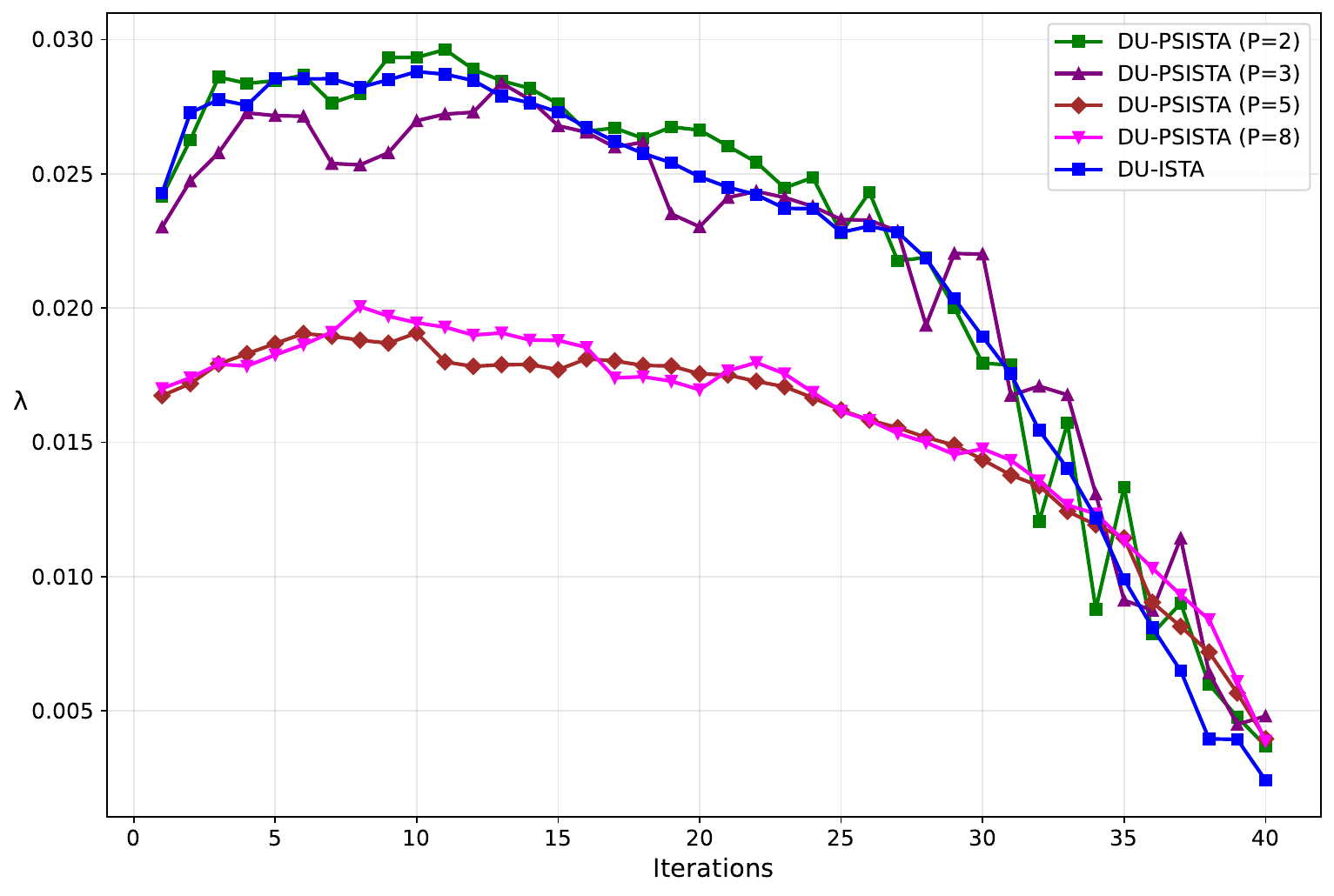}
  \caption{Learned thresholding parameters $\lambda_t$ for different period $P$ in {DU-PSISTA} for large system ($n=1024, m=512, l=256$).}
  \label{p_large_lambdas}
\end{figure}
Finally, we observe the values of learned parameters in the proposal DU-PSISTA 
and verify their consistency with the convergence analysis.

Figures \ref{p_large_etas} and \ref{p_large_lambdas} show the examples of the learned step sizes $\eta_t$ and thresholding parameters $\lambda_t$ for the large system with $l=256$ and different periods $P$, respectively.
From Fig.~\ref{p_large_etas}, it can be observed that 
the step sizes for each period take larger values when $(t-1)\bmod P=0$ than when $(t-1)\bmod P\neq0$.
For example, when $P=5$, it tends to take on large values when it is a multiple of 5 plus 1.
Figure~\ref{p_large_lambdas} indicates that 
$\lambda_t$ for $P=2, 3$ take values quite close to those of DU-ISTA's $\lambda'$.
We can see that, from Fig.~\ref{jikken1}, 
DU-PSISTA with $P=2, 3$ achieves recovery accuracy comparable to DU-ISTA.
These results are consistent with the discussion in the proof of Corollary~\ref{corr:improve}: 
the step sizes for iterations $(t-1)\bmod P=0$ take larger values 
and those for iterations $(t-1)\bmod P\neq0$ take smaller values
to reduce the difference in error between the proposed DU-PSISTA and DU-ISTA.
Moreover, the statistic $\tilde{\lambda}$ should be satisfy $\tilde{\lambda}\approx\lambda'$ to achieve the same recovery accuracy as DU-ISTA.
This is consistent with the discussion in \eqref{eq:cd} 
that the larger $l$ is, the larger $C/D$ becomes, allowing for a larger step size.
The values of $\lambda_t$ for $l=32, 64$ are relatively close to those of DU-ISTA's $\lambda'$, 
indicating the consistency with the results in Fig.~\ref{jikken4} that the proposed method can achieve recovery accuracy comparable to DU-ISTA when $l=32, 64$.
% In addition, the actual learned parameter values tend to have smaller step sizes when using sketching, 
% which is consistent with Collorary~\ref{corr:delay}.

\begin{figure}[t]
  \centering
  \includegraphics[width=1\linewidth]{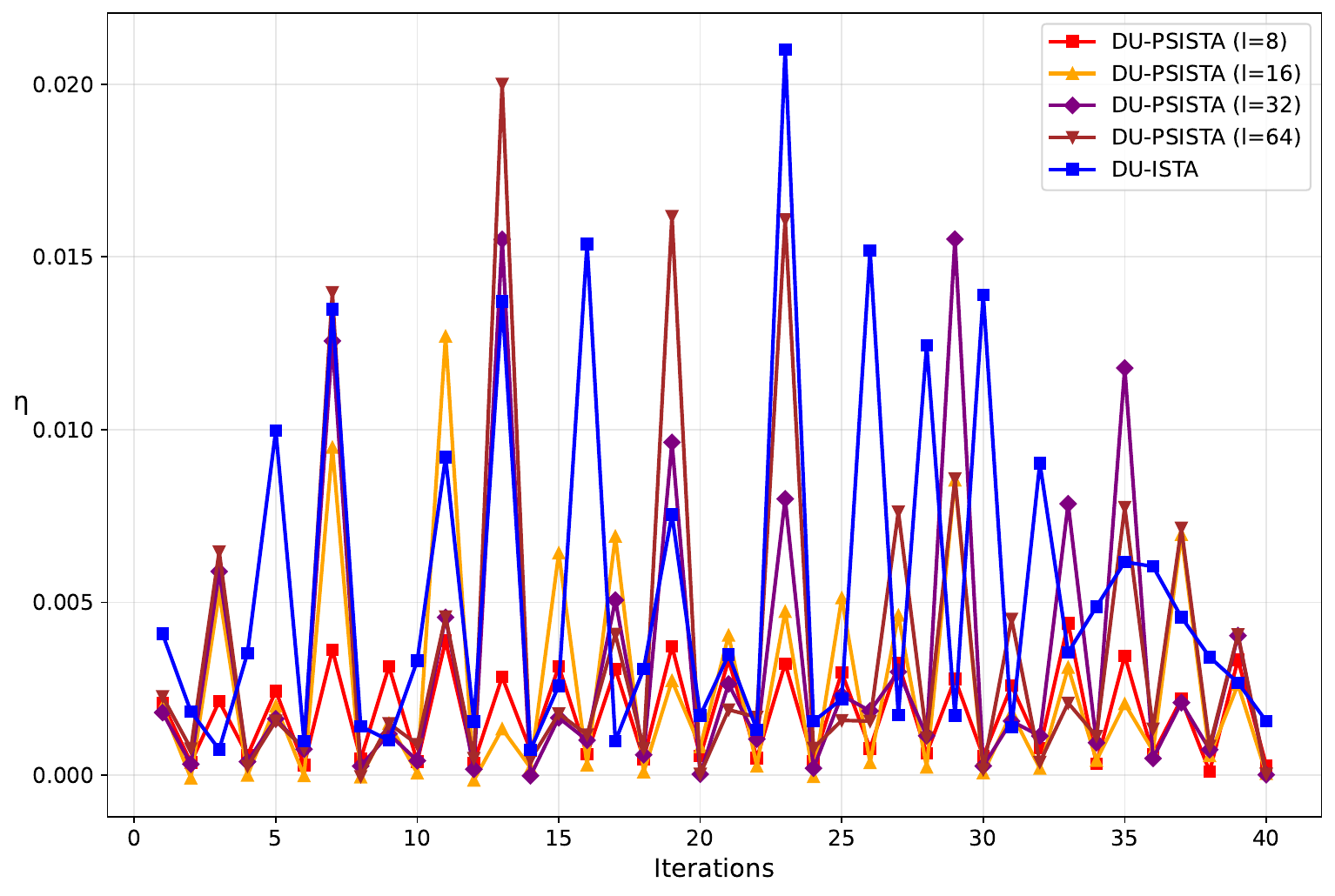}
  \caption{Learned step sizes $\eta_t$ for different sketch size $l$ in {DU-PSISTA} for small system ($n=256, m=128, P=2$).}
  \label{l_small_etas}
\end{figure}
\begin{figure}[t]
  \centering
  \includegraphics[width=1\linewidth]{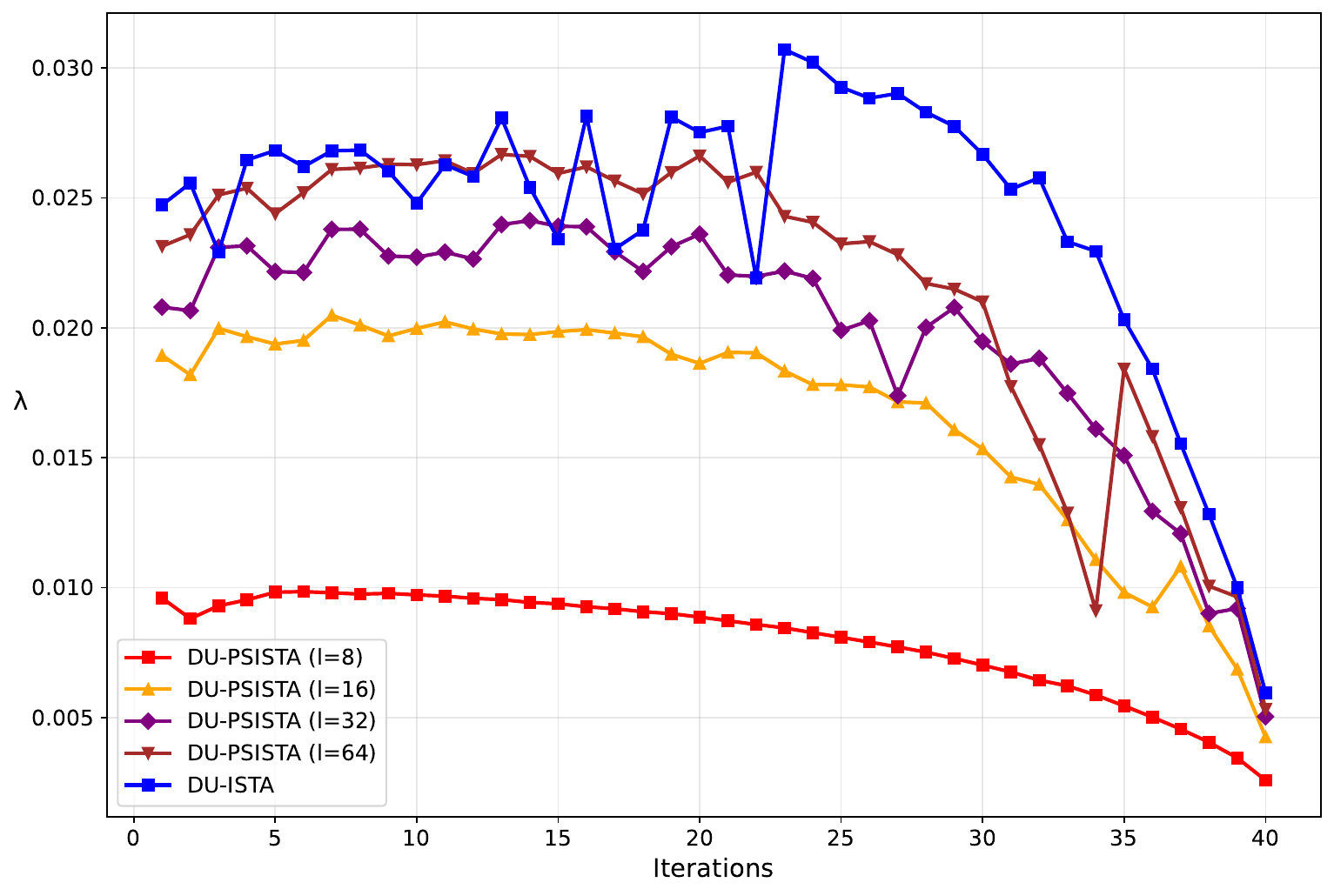}
  \caption{Learned thresholding parameters $\lambda_t$ for different sketch size $l$ in {DU-PSISTA} for small system ($n=256, m=128, P=2$).}
  \label{l_small_lambdas}
\end{figure}
Figures \ref{l_small_etas} and \ref{l_small_lambdas} show the examples of the learned step sizes $\eta_t$ and thresholding parameters $\lambda_t$ for the small system with $P=2$ and different sketch size $l$, respectively.
As discussed above, 
the magnitude of the step size value in each iteration tends to depend on the choice of $P$. 
Therefore, in this case ($P=2$), it can be observed that for any value of $l$, the step sizes tend to take on larger values in the same iteration.
Moreover, we can see that, in each iteration, 
the magnitude for larger $l$ tends to be larger than that for smaller $l$.

In summary, there is no contradiction between the conditions for the learnable parameters of DU-PSISTA revealed by convergence analysis 
and the values of the parameters actually learned in the computer experiments. 
Therefore, the validity of the convergence analysis is supported.

\section{Conclusions}
\label{sec:conclusion}
In this paper, we provided DU-PSISTA as a method aiming to achieve both computational efficiency and estimation accuracy for the recovery problem of high-dimensional and sparse signals.
The proposed algorithm is characterized by the ability to reduce computational cost while maintaining accuracy by periodically switching between original  and sketched gradient update.
Furthermore, by applying deep unfolding to this method, we extended it so that step sizes and 
thresholding parameters in each iteration can be learned, enabling data-driven optimization of hyperparameters that previously had to be manually adjusted.
We showed that the proposed method achieves a linear-type contraction to a neighborhood of the true sparse signal 
with properly selected parameters.
The analysis provides an explanation for the effectiveness of 
mixing original and sketched gradient updates to improve recovery accuracy.

The effectiveness of the proposed DU-PSISTA was verified through computer experiments conducted under various system sizes, sketch sizes, and period settings.
As a result, it was confirmed that, especially with settings that appropriately balance original and sketched gradient updates, 
it is possible to reduce computational complexity by 32 \% while maintaining recovery performance comparable to the conventional DU-ISTA.
The learned parameters in the experiments also show consistency with the theoretical analysis, supporting the validity of the convergence analysis.

% As future work, we would like to aim for further improvement in computational efficiency and recovery accuracy by incorporating other structured sketching matrices.
The periodical application of sketching and introduction of deep unfolding can be combined with other sparse signal recovery algorithms such as FISTA.
As future work, 
the introduction of architecture search to the proposed method 
is one of the other challenging directions for improving performance.
We used periodic structure for the proposed method in this paper, 
but it is possible to introduce other mixing schedules such as adaptive or dynamic scheduling.
It has possibility to adaptively adjust the computational complexity and recovery accuracy 
according to the system environment and recovery accuracy requirements.
Furthermore, the adaptive control of parameters or schedules instead of pre-computation of deep unfolding 
may be another challenging direction for enhancing practicality.

\appendix

% \section{Derivation of Theorem~\ref{thm:error_bound}}
The detailed proof of Theorem~\ref{thm:error_bound} is provided in the following.
We use the following notation that is introduced in \cite{chen2018theoretical}:
\begin{definition}
  A set $\partial \ell_1(\bm{x}) \ (\bm{x}=[x_1,\ldots, x_n]^\mathrm{T})$ is defined as below 
  as a subgradient of $\|\bm{x}\|_1$:
  \begin{equation}
    \partial \ell_1(x_i)=\begin{cases}            \{\mathrm{sign}(x_i)\}, & \mathrm{if} \ x_i\neq0, \\            [-1,1], & \mathrm{if}~x_i=0.        \end{cases}
  \end{equation}
  \label{def:subgradient}
\end{definition}

First, we show that $x_i^{(t)}=0$ for $\forall i\notin\Omega $ (where $x_i^\star=0$) by induction.
Let $x_i^{(0)}=0, \ \forall \ i\notin\Omega $. 
Then, it holds for $t=0$.
We then assume that $x_i^{(t)}=0, \ \forall \ i\notin\Omega $. 
From the update equation \eqref{eq:update}, 
the components $x_i^{(t+1)} \ (i=1,\ldots,n)$ of $\bm{x}^{(t+1)}$ are expressed as
\begin{align}
  x_i^{(t+1)} &= S_{\lambda_t}(x_i^{(t)}-\eta_t\bm{A}_{t,i}^\mathrm{T}\bm{A}_t(\bm{x}^{(t)}-\bm{x}^\star) + \eta_t\tilde{\bm{A}}_{t,i}^\mathrm{T}\bm{w})\\
  &= S_{\lambda_t}(-\eta_t\sum_{j\in\Omega }\bm{A}_{t,i}^\mathrm{T}\bm{A}_{t,j}(x_j^{(t)}-x_j^\star) + \eta_t\tilde{\bm{A}}_{t,i}^\mathrm{T}\bm{w}).
\end{align}
Here, we use the fact that $\forall i'\notin \Omega $, $x_{i'}^{(t)}-x_{i'}^\star=0-0=0$. 
By Assumption \ref{assum:lambda}, we have $x_i^{(t+1)}=0$.
Therefore, by induction, we have $x_i^{(t)}=0$ for $\forall i\notin\Omega  \ (x_i^\star=0)$.

Next, we show about the components $\forall i \in\Omega $ where $x_i^\star\neq0$.
The component $x_i^{(t+1)}$ can be written as follows:
$x_i^{(t+1)} = S_{\lambda_t}(x_i^{(t)}-\eta_t\bm{A}_{t,i}^\mathrm{T}(\bm{A}_t\bm{x}^{(t)}-\bm{y}_t)) \in x_i^{(t)}-\eta_t\bm{A}_{t,i}^\mathrm{T}(\bm{A}_t\bm{x}^{(t)}-\bm{y}_t) - \lambda_t\partial\ell_1(x_i^{(t+1)})$.
% \begin{align}
% x_i^{(t+1)} &= S_{\lambda_t}(x_i^{(t)}-\eta_t\bm{A}_{t,i}^\mathrm{T}(\bm{A}_t\bm{x}^{(t)}-\bm{y}_t))\\       & \in x_i^{(t)}-\eta_t\bm{A}_{t,i}^\mathrm{T}(\bm{A}_t\bm{x}^{(t)}-\bm{y}_t) - \lambda_t\partial\ell_1(x_i^{(t+1)}).
% \end{align}
The vector $\bm{x}^{(t+1)}$ can be written by summarizing the components as follows:
\begin{align}
\bm{x}^{(t+1)}&=       \bm{x}^{(t)}-\eta_t\bm{A}_t^\mathrm{T}(\bm{A}_t\bm{x}^{(t)}-\bm{y}_t)-\lambda_t\bm{p}^{t+1} \\ 
&= \bm{x}^{(t)}-\eta_t\bm{A}_t^\mathrm{T}\bm{A}_t\bm{x}^{(t)}+\eta_t\bm{A}_t^\mathrm{T}\bm{A}_t\bm{x}^\star \nonumber \\
&\quad+\eta_t\tilde{\bm{A}}_t^\mathrm{T}\bm{w} -\lambda_t\bm{p}^{t+1},
\end{align}
where $\bm{p}^{t+1}=[\bm{p}^{t+1}_1,\ldots,\bm{p}^{t+1}_n]^\mathrm{T}$ and 
each element satisfies $p^{t+1}_i\in\partial\ell_1(x_i^{(t+1)})$.

Consider the norm of the error between the true signal and the estimated signal.
Here, since we have shown that $x_i^{(t)}=0$ for $\forall i\notin\Omega $, 
it holds that $\|\bm{x}^{(t)}-\bm{x}^\star\|=\|\bm{x}^{(t)}-\bm{x}^\star\|_\Omega  \ \forall \ t$,
where $\|\cdot\|_\Omega $ is the $\ell_2$ norm that is calculated by taking only the components $i\in\Omega $.
By Definition \ref{def:subgradient}, it holds that $\|\bm{p}^t\|_\Omega \leq\sqrt{s} \ \forall \ t$.
We define $h^{t+1}=\|\bm{x}^{(t+1)}-\bm{x}^\star\|$ and it can be expanded as follows:
\begin{align}
h^{t+1}&=\|\bm{x}^{(t+1)}-\bm{x}^\star\|_\Omega  \nonumber\\ 
&= \|\bm{x}^{(t)}-\eta_t\bm{A}_t^\mathrm{T}\bm{A}_t\bm{x}^{(t)}+\eta_t\bm{A}_t^\mathrm{T}\bm{A}_t\bm{x}^\star \nonumber \\
&\quad+\eta_t\tilde{\bm{A}}_t^\mathrm{T}\bm{w} -\lambda_t\bm{p}^{t+1} - \bm{x}^\star\|_\Omega  \nonumber \\ 
&= \|(\bm{I}-\eta_t\bm{A}_t^\mathrm{T}\bm{A}_t)(\bm{x}^{(t)}-\bm{x}^\star)+\eta_t\tilde{\bm{A}}_t^\mathrm{T}\bm{w} -\lambda_t\bm{p}^{t+1}\|_\Omega  \nonumber \\ 
&\leq \|\bm{I}-\eta_t\bm{A}_t^\mathrm{T}\bm{A}_t\|_\Omega \|\bm{x}^{(t)}-\bm{x}^\star\|_\Omega  \nonumber \\
&\quad+\eta_t\|\tilde{\bm{A}}_t^\mathrm{T}\bm{w}\|_\Omega  +\lambda_t\|\bm{p}^{t+1}\|_\Omega  \nonumber \\ 
&\leq \|\bm{I}-\eta_t\bm{A}_t^\mathrm{T}\bm{A}_t\|_\Omega h^t+\eta_t\|\tilde{\bm{A}}_t^\mathrm{T}\bm{w}\|_\Omega  +\lambda_t\sqrt{s} \label{eq:h_bound_t} \\ 
&\leq \ldots \nonumber \\ 
&\leq \left(\prod_{t'=1}^t\|\bm{I}-\eta_{t'}\bm{A}_{t'}^\mathrm{T}\bm{A}_{t'}\|_\Omega \right)h^1 \nonumber \\
&\quad+ \sum_{t'=1}^t\left(\eta_{t'}\|\tilde{\bm{A}}_{t'}^\mathrm{T}\bm{w}\|_\Omega +\lambda_{t'}\sqrt{s}\right) \nonumber \\
&\quad\times\left(\prod_{t''=t'+1}^t\|\bm{I}-\eta_{t''}\bm{A}_{t''}^\mathrm{T}\bm{A}_{t''}\|_\Omega \right).
% &= \|\bm{x}^{(t)}-\eta_t\bm{A}_t^\mathrm{T}\bm{A}_t\bm{x}^{(t)}+\eta_t\bm{A}_t^\mathrm{T}\bm{A}_t\bm{x}^\star+\eta_t\tilde{\bm{A}}_t^\mathrm{T}\bm{w} -\lambda_t\bm{p}^{t+1} - \bm{x}^\star\|_\Omega  \\ = \|(\bm{I}-\eta_t\bm{A}_t^\mathrm{T}\bm{A}_t)(\bm{x}^{(t)}-\bm{x}^\star)+\eta_t\tilde{\bm{A}}_t^\mathrm{T}\bm{w} -\lambda_t\bm{p}^{t+1}\|_\Omega  \\ \leq \|\bm{I}-\eta_t\bm{A}_t^\mathrm{T}\bm{A}_t\|_\Omega \|\bm{x}^{(t)}-\bm{x}^\star\|_\Omega +\eta_t\|\tilde{\bm{A}}_t^\mathrm{T}\bm{w}\|_\Omega  +\lambda_t\|\bm{p}^{t+1}\|_\Omega  \\ \leq \|\bm{I}-\eta_t\bm{A}_t^\mathrm{T}\bm{A}_t\|_\Omega h^t+\eta_t\|\tilde{\bm{A}}_t^\mathrm{T}\bm{w}\|_\Omega  +\lambda_t\sqrt{s} \\ \leq \ldots \\ \leq \left(\prod_{t'=1}^t\|\bm{I}-\eta_{t'}\bm{A}_{t'}^\mathrm{T}\bm{A}_{t'}\|_\Omega \right)h^1 + \sum_{t'=1}^t\left(\eta_{t'}\|\tilde{\bm{A}}_{t'}^\mathrm{T}\bm{w}\|_\Omega +\lambda_{t'}\sqrt{s}\right)\left(\prod_{t''=t'+1}^t\|\bm{I}-\eta_{t''}\bm{A}_{t''}^\mathrm{T}\bm{A}_{t''}\|_\Omega \right)
\label{eq:h_bound}
\end{align}

We expand the second term in the right-hand side of \eqref{eq:h_bound}.
By Assumption \ref{assum:xw} and Definition \ref{def:max}, we have
\begin{align}
\|\tilde{\bm{A}}_{t'}^\mathrm{T}\bm{w}\|_\Omega  \leq \begin{cases}\sqrt{s}C\epsilon_w, & \mathrm{if} \ (t-1)\bmod{P}=0, \\ \sqrt{s}D\epsilon_w, & \mathrm{otherwise}.\end{cases}
\end{align}
The second term in the right-hand side of \eqref{eq:h_bound} can be expanded as 
\begin{align}
&\sum_{t'=1}^t\left(\eta_{t'}\|\tilde{\bm{A}}_{t'}^\mathrm{T}\bm{w}\|+\lambda_{t'}\sqrt{n}\right)\left(\prod_{t''=t'+1}^t\|\bm{I}-\eta_{t''}\bm{A}_{t''}^\mathrm{T}\bm{A}_{t''}\|\right) \nonumber\\ 
&<\sum_{t'\in[t], (t'-1)\bmod{P}=0}\left(\eta_{t'}\sqrt{s}C\epsilon_w+\lambda_{t'}\sqrt{s}\right) \nonumber \\
&\quad\times\prod_{t''=t'+1}^t\|\bm{I}-\eta_{t''}\bm{A}_{t''}^\mathrm{T}\bm{A}_{t''}\|_\Omega  \nonumber \\
&\quad+\sum_{t'\in[t], (t'-1)\bmod{P}\neq0}\left(\eta_{t'}\sqrt{s}D\epsilon_w+\lambda_{t'}\sqrt{s}\right) \nonumber \\
&\quad\times\prod_{t''=t'+1}^t\|\bm{I}-\eta_{t''}\bm{A}_{t''}^\mathrm{T}\bm{A}_{t''}\|_\Omega .
% &= \sqrt{s}C\sigma\left(\sum_{t'\in[t], t'\mod{P}=1}\eta_{t'}\right)+\sqrt{s}D\sigma\left(\sum_{t'\in[t], t'\mod{P}\neq1}\eta_{t'}\right)+\sqrt{s}\left(\sum_{t'\in[t]} \lambda_{t'}\right)\\ < \sqrt{s}C\sigma\tilde{\eta}_C+\sqrt{s}D\sigma\tilde{\eta}_D+\sqrt{s}\tilde{\lambda}
\end{align}
where $[t]$ is the set of all integers from $1$ to $t$.
By using Assumptions \ref{assum:lambda} and \ref{assum:eta}, we can further expand the second term in the right-hand side of \eqref{eq:h_bound} as follows:
\begin{align}
  &\sum_{t'=1}^t\left(\eta_{t'}\|\tilde{\bm{A}}_{t'}^\mathrm{T}\bm{w}\|+\lambda_{t'}\sqrt{n}\right)
  \left(\prod_{t''=t'+1}^t\|\bm{I}-\eta_{t''}\bm{A}_{t''}^\mathrm{T}\bm{A}_{t''}\|\right) \nonumber\\ 
  &<\sum_{t'\in[t], (t'-1)\bmod{P}=0}\left(\eta_{t'}\sqrt{s}C\epsilon_w+\lambda_{t'}\sqrt{s}\right) \nonumber \\
  &\quad+\sum_{t'\in[t], (t'-1)\bmod{P}\neq0}\left(\eta_{t'}\sqrt{s}D\epsilon_w+\lambda_{t'}\sqrt{s}\right) \nonumber\\ 
  &= \sqrt{s}C\epsilon_w\left(\sum_{t'\in[t], (t'-1)\bmod{P}=0}\eta_{t'}\right) \nonumber \\
  &\quad+\sqrt{s}D\epsilon_w\left(\sum_{t'\in[t], (t'-1)\bmod{P}\neq0}\eta_{t'}\right)+\sqrt{s}\left(\sum_{t'\in[t]} \lambda_{t'}\right)\nonumber \\ 
  &< \sqrt{s}C\epsilon_w\tilde{\eta}_C+\sqrt{s}D\epsilon_w\tilde{\eta}_D+\sqrt{s}\tilde{\lambda}.
  \label{eq:h_bound_expanded}
\end{align}
Therefore, Theorem~\ref{thm:error_bound} is proved.

\section*{Acknowledgments}
This work was supported by 
JSPS KAKENHI Grant-in-Aid for Young Scientists Grant Number JP23K13334 (to A. Nakai-Kasai) 
and Scientific Research(A) Grant Number JP25H01111 (to T. Wadayama).

\pagestyle{plain}
\printbibliography

 

% \end{IEEEbiography}



\begin{IEEEbiography}[{\includegraphics[width=1in,height=1.25in,clip,keepaspectratio]{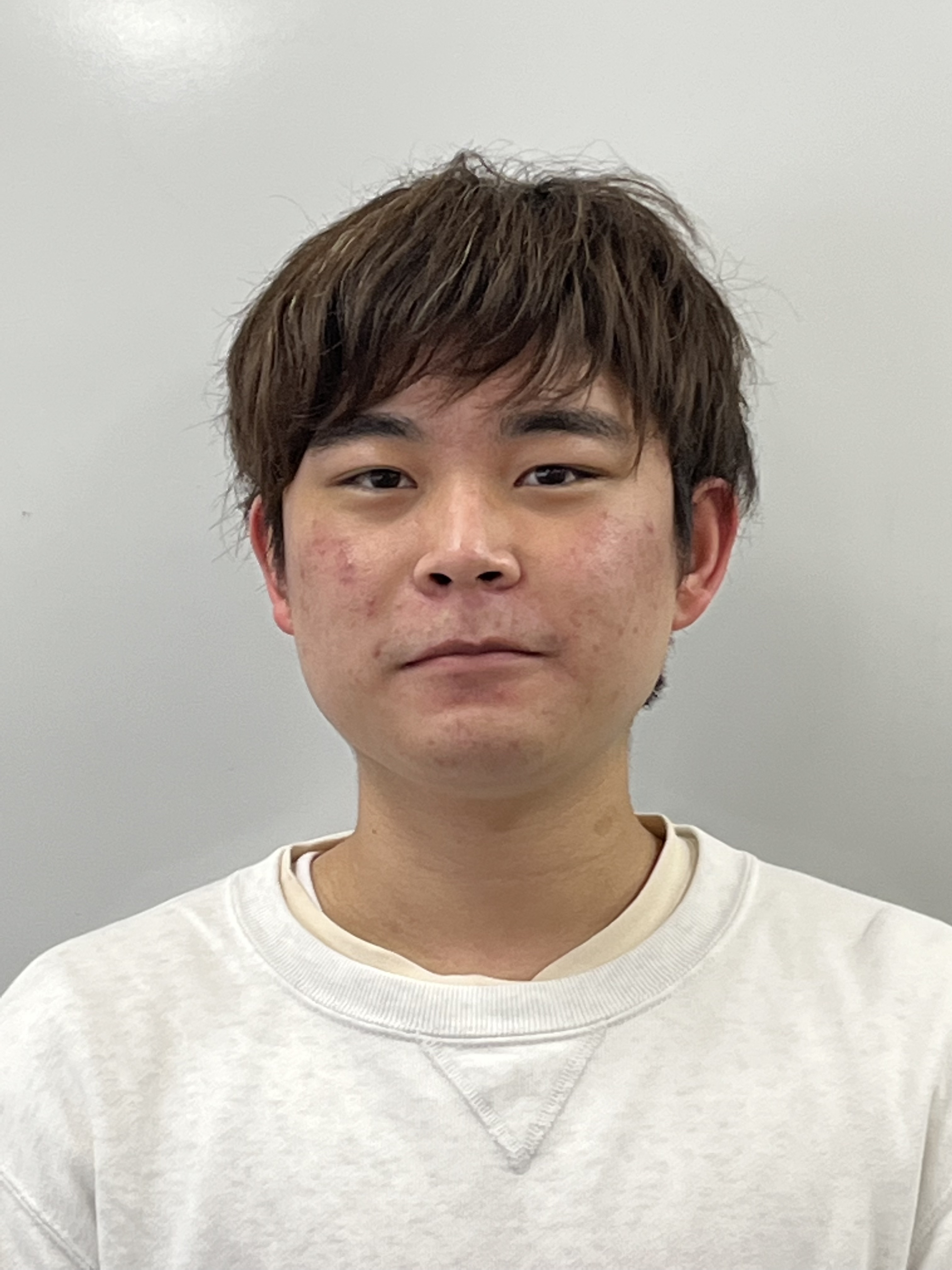}}]{Tatsuki Tokumura}
  received B.E. degree from Nagoya Institute of Technology, Japan, in 2024.
  He is currently a Master's student at Graduate School of Engineering, Nagoya Institute of Technology.
  His research interests include signal processing for wireless communications and machine learning.
\end{IEEEbiography}
\begin{IEEEbiography}[{\includegraphics[width=1in,height=1.25in,clip,keepaspectratio]{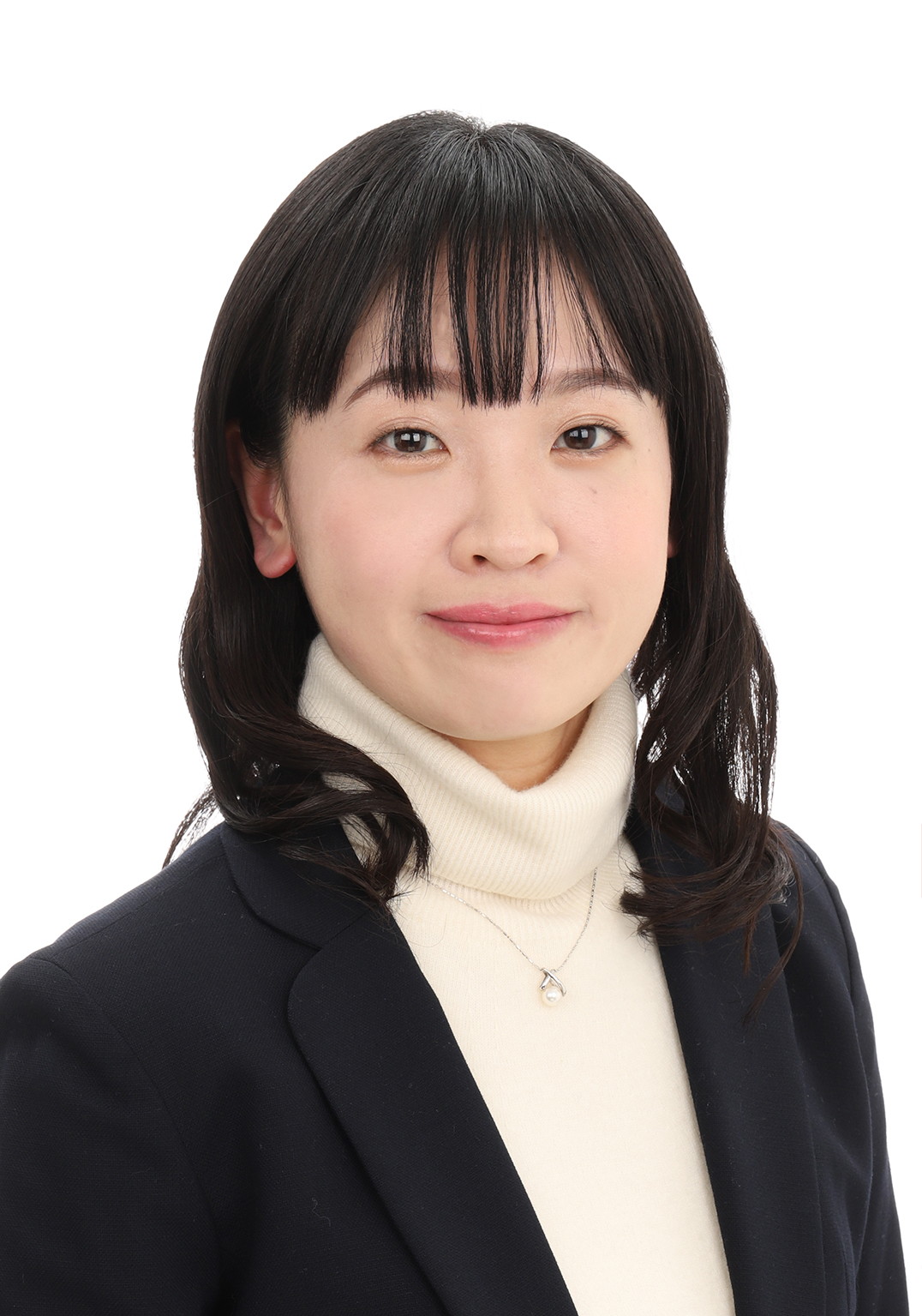}}]{Ayano Nakai-Kasai}
    received the bachelor’s degree 
    in engineering, the master’s degree in informatics,
    and Ph.D. degree in informatics from Kyoto University, Kyoto, Japan, in 2016, 2018, and 2021,
    respectively. She is currently an Assistant Professor
    at Graduate School of Engineering, Nagoya Institute of Technology. Her research interests include
    signal processing, wireless communication, and machine learning. She received the Young Researchers’
    Award from the Institute of Electronics, Information 
    and Communication Engineers in 2018 and APSIPA 
    ASC 2019 Best Special Session Paper Nomination Award. She is a member
    of IEEE and IEICE.
\end{IEEEbiography}
\begin{IEEEbiography}[{\includegraphics[width=1in,height=1.25in,clip,keepaspectratio]{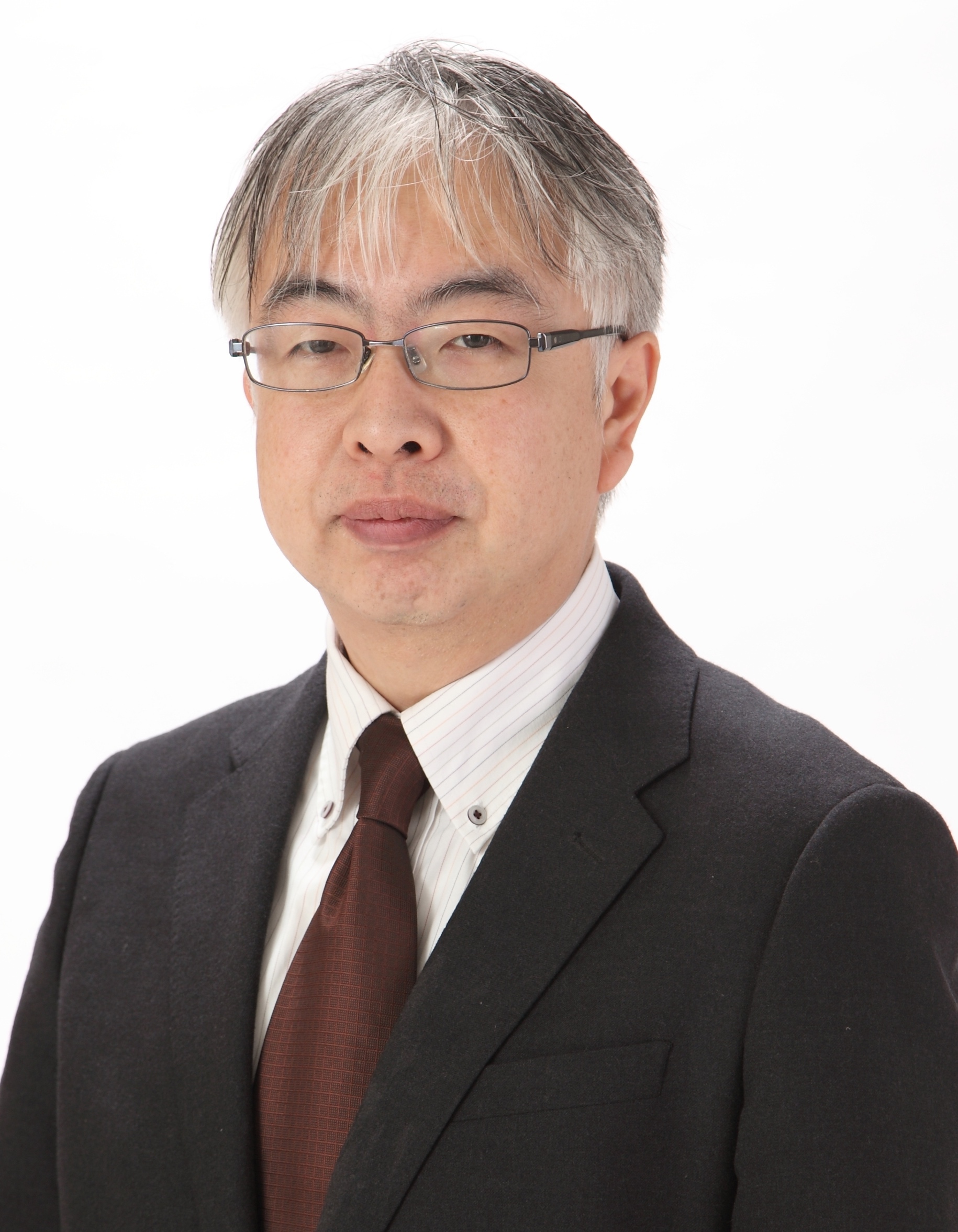}}]{Tadashi Wadayama}
    (M'96) was born in Kyoto, Japan, on May 9,1968.
    He received the B.E., the M.E., and the D.E. degrees from Kyoto Institute of Technology in 1991, 1993 and 1997, respectively.
    On 1995, he started to work with Faculty of Computer Science and System Engineering, Okayama Prefectural University as a research associate.
    From April 1999 to March 2000, he stayed in Institute of Experimental Mathematics, Essen University (Germany) as a visiting researcher.
    On 2004, he moved to Nagoya Institute of Technology as an associate professor. Since 2010, he has been a full professor of Nagoya Institute of Technology.
    His research interests are in coding theory, information theory, and signal processing for wireless communications.
    He is a member of IEEE and a senior member of IEICE.
\end{IEEEbiography}

\vfill

\end{document}